\documentclass{article}

\usepackage{geometry}
\usepackage[round]{natbib}

\usepackage{mathtools}
\usepackage{amssymb}
\usepackage{amsthm}

\usepackage{tikz}
\usepackage{xcolor,graphicx}
\usepackage{enumitem}
\usepackage{placeins}  
\usepackage{booktabs}  
\usepackage[labelfont=bf,font=small,width=.9\textwidth]{caption}  

\usepackage[hidelinks]{hyperref}
\usepackage{bookmark}

\newcommand{\diff}{\mathrm{d}}
\newcommand{\E}{\mathbb{E}} 
\renewcommand{\P}{\mathbb{P}}
\newcommand{\R}{\mathbb{R}}
\newcommand{\e}{\mathrm{e}}
\newcommand{\indic}{\mathbf{1}}

\definecolor{couleurrouge}{HTML}{cf021b}
\definecolor{couleurbleu}{HTML}{4988d5}

\DeclarePairedDelimiter{\Angle}{\langle}{\rangle}
\DeclarePairedDelimiter{\abs}{|}{|}
\renewcommand{\angle}{\Angle}

\DeclareMathOperator{\diag}{Diag}

\newtheorem{theorem}{Theorem}

\newtheorem{proposition}{Proposition}

\allowdisplaybreaks

\usepackage{authblk}

\newlength{\affilskip}
\author[1]{Félix Foutel-Rodier}
\author[2]{Arthur Charpentier}
\author[2]{Hélène Guérin}

\affil[1]{
    Department of Statistics, University of Oxford
    \vspace{\affilskip}
}
\affil[2]{
    Département de Mathématiques, Université du Québec à Montréal
    \vspace{\affilskip}
}

\title{Optimal Vaccination Policy to Prevent Endemicity\\A Stochastic Model}

\begin{document}

\maketitle

\begin{abstract}

We examine here the effects of recurrent vaccination and waning immunity
on the establishment of an endemic equilibrium in a population. An
individual-based model that incorporates memory effects for
transmission rate during infection and subsequent immunity is
introduced, considering stochasticity at the individual level. By
letting the population size going to infinity, we derive a set of
equations describing the large scale behavior of the epidemic.
The analysis of the model's equilibria reveals a criterion for the
existence of an endemic equilibrium, which depends on the rate of
immunity loss and the distribution of time between booster doses. The
outcome of a vaccination policy in this context is influenced by the
efficiency of the vaccine in blocking transmissions and the distribution
pattern of booster doses within the population. Strategies with evenly
spaced booster shots at the individual level prove to be more effective
in preventing disease spread compared to irregularly spaced boosters, as
longer intervals without vaccination increase susceptibility and
facilitate more efficient disease transmission. We provide an expression
for the critical fraction of the population required to adhere to the
vaccination policy in order to eradicate the disease, that resembles a
well-known threshold for preventing an outbreak with an imperfect
vaccine. We also investigate the consequences of unequal vaccine access
in a population and prove that, under reasonable assumptions, fair vaccine
allocation is the optimal strategy to prevent endemicity.


\end{abstract}

\section{Introduction}

In epidemiology, a disease is called endemic if it persists in a
population over a long period of time. Many diseases are endemic in some
parts of the world, including for instance malaria and tuberculosis
\citep{hay2009world, oliwa2015tuberculosis}, and several studies have
proposed that endemicity is a likely outcome for the recent COVID-19
epidemic \citep{antia2021transition, lavine2021immunological}, as is
currently the case for other human coronavirus-induced diseases
\citep{su2016epidemiology}. The persistence of a disease in a population
can incur a large cost for society and endemic diseases are responsible
for a large share of the deaths from communicable diseases every year.
Understanding the mechanisms underlying the establishment of such an
endemic state and how to control it is therefore of great public health
importance. Prophylactic vaccination, when available, is a common and
efficient way to mitigate the spread of diseases
\citep{plotkin2005vaccines, rashid2012vaccination}. If the vaccine blocks
part of the transmissions, a high enough vaccine coverage can prevent
self-sustained transmissions in the population, leading to a so-called
herd immunity \citep{anderson1985vaccination, fine2011herd,
randolph2020herd}. This phenomenon has been the subject of a large body
of work in the mathematical modeling literature, aimed at informing
policy-makers on the effectiveness of a vaccination campaign and at
developing a theoretical understanding of the epidemiological
consequences of vaccination. An important achievement of these studies is
the derivation of an expression for the critical vaccine coverage
required to eradicate a disease, under various scenarios of increasing
complexity \citep{anderson1982directly, farrington2003vaccine,
magpantay2017vaccine, delmas2022infinite}. However, the bulk of this work
pertains to vaccines providing life-long (or slowly waning) immunity and
administrated at birth or at a single point in time.
Although these assumptions might represent adequately many situations 
(including for instance childhood diseases), infection by some pathogens
and vaccines are known to provide no or temporary immunity
\citep{vynnycky1997natural, rts2015efficacy, stein2023past}. An
important motivating example for our work is the recent COVID-19
epidemic, for which reinfections after either primary infection or
vaccination have been reported \citep{stein2023past}, and for which direct
measurements of several components of adaptive immunity suggest that part
of it is waning \citep{shrotri2022spike, lin2022longitudinal}. The
understanding of the impact of vaccination under such short-lived
immunity remains limited and motivates further theoretical developments.

In this work, we consider a pathogen for which a vaccine that blocks
transmissions is available but with an immunity that wanes with time,
both for individuals infected and vaccinated. Although our main motivation is
COVID-19, we consider a generic disease with these two features.
As the immunity conferred by the vaccine is temporary, the effect of a
single vaccination rapidly fades and herd immunity can only be achieved
(and thus an endemic state prevented) if individuals are vaccinated
recurrently \citep{randolph2020herd}. However, even recurrent vaccination
might fail to provide herd immunity. Under recurrent vaccination, the
level of immunity in the population is shaped by two antagonistic forces:
boosting due to vaccine injections and re-exposure to the pathogen, and
waning due to decay in circulating antibody levels and/or memory cells. If
vaccination is too scarce or immunity decays too rapidly, the
vaccine-induced immunity might not block enough transmissions to prevent
the disease from spreading in the population and reaching endemicity.
What drives the outcome of a vaccination policy is therefore a complex
interplay between the transmissibility of the disease, the waning of the
immunity (which sets up the time scale after which reinfection can
occur) and the frequency of immune boosting by vaccines. 
We investigate this effect by constructing an epidemic model that
incorporates both waning immunity and recurrent vaccination, and by 
analysing how these two components interact to determine the long-term
establishment of the disease. 

In standard SIR-type models, waning immunity can be modeled by letting
the infected individuals go
back to a susceptible state, either directly after the infection as in
the SIS model, or after a temporary immune period as in the SIRS model
\citep{brauer2019mathematical}. Extensions of these models where the
duration of the immune period is fixed or has a general distribution have
also been proposed, for instance in \cite{hethcote1981nonlinear,
cooke1996analysis, taylor2009sir, bhattacharya2012time}. In such models,
immunity is lost instantaneously as individuals go from being fully
protected (in the $R$ state) to being fully susceptible to the disease
(in the $S$ state). Some studies consider a more gradual loss of immunity 
by adding one or several intermediate compartments with partial immunity,
often denoted by $W$ (for waning) \citep{lavine2011natural,
carlsson2020modeling}. In our work, we will model the decay of immunity by
tracking for each individual a \emph{susceptibility} giving the
probability of being reinfected upon exposure to the pathogen. Waning
immunity is modeled by having the susceptibility increase with time
following an infection or vaccination, with no further assumption. This
approach can account for the situations described above, where each
individual is in one of finitely many immune states ($R$, $S$, $W$), but
also for a continuous loss of immunity. The idea of modeling a
susceptibility dates back to the endemic models of Kermack and McKendrick
\citep{kermack32, kermack33}, see also \cite{inaba01, breda12} for modern
formulations, and is also reminiscent of existing works describing
immunity as a continuous variable \citep{white1998microparasite,
diekmann2018waning, barbarossa2015immuno, martcheva2006epidemic}.
Modeling immunity through an abstract susceptibility is a
phenomenological approach, but mechanistic approaches have also been
proposed. These require to model explicitly for each individual some
components of the immune system (T-cells, B-cells, antibodies, cytokines)
and their interaction with the pathogen, as for instance in
\cite{heffernan2008host, heffernan2009implications, goyal2020potency,
neant2021modeling}. Such an approach is both more realistic and opens the
possibility of being calibrated using clinical data
\citep{lin2022longitudinal}, but adds a new layer of complexity (the
within-host dynamics) which can be cumbersome for theory purpose. We will
think of our susceptibility as aggregating the effect of this complicated
within-host process.

The effect of immune boosting through recurrent vaccination has also
drawn attention from modelers \citep{arino2003global, lavine2011natural,
carlsson2020modeling, leung2018infection}. Let us note that, in the
vaccination strategy that we consider, individuals are vaccinated
recurrently during their lifetime, each at different moments. The name
``continuous vaccination'' has also been proposed for this type of
vaccination \citep{liu2008svir, li2011sir}. This is different for instance
of the ``pulse vaccination strategy'', considered in \cite{agur1993pulse}
in the context of measles, and that has received further attention
\citep{liu2008svir, li2011sir}. In this strategy, at several fixed moments
a fraction of the population is vaccinated, all individuals in a given
vaccination pulse receiving their dose at the same time. Typically,
recurrent boosting is modeled by letting individuals get vaccinated at a
given rate (that can depend on age or other factors) in which case they
are moved to a compartment with a reduced susceptibility. A notable
difference with our work is that, in our model, the vaccination rate
depends on the time elapsed since the previous vaccination. This reflects
the fact that, at the microscopic level, we let the period of time
between two consecutive vaccinations have a general (non-exponential)
distribution. We have several motivations for relaxing the usual constant
rate assumption. First, the resulting dynamics is much richer and
complex. It encompasses more realistic situations that cannot be modeled
using a constant rate, for instance the enforcement of a minimal duration
between two vaccine doses, or the existence of a typical duration between
two doses, leading to a peak in the distribution of this duration.
Second, summarizing the effect of vaccination and waning immunity by a
small number of parameters (the transition rates between the various
immune compartments) obscures the role played by the exact shapes of the
immunity decay and of the distribution between vaccine doses and leads to
quite opaque expressions. For instance, if immunity following vaccination
first plateaus and is then lost rapidly around a typical time, we expect
a population where boosting occurs right before this time to build a much
stronger immunity than if boosting occurred right after this time,
although there is only a minimal variation in the overall vaccination
rate. This type of effect cannot be studied by assuming that all
durations have exponential distributions. From a mathematical point of
view, this more general model requires to work with partial differential
equations describing the ``age structure'' of the population rather than
with more usual sets of ordinary differential equations. It is always
interesting to note that similar age-structured epidemic models were used
as early as in the foundational work of Kermack and McKendrick
\citep{kermack1927contribution}, and that the compartmental SIR model was
only introduced as a particular case of this more general dynamics.

Our model has one last specificity compared to more classical
approaches based on compartments. It is formulated as an
individual-based stochastic model from which we derive a set of
deterministic equations describing its scaling limit (as the population
size goes to infinity). Modeling the population at the microscopic level
rather than directly at the continuum gives a better understanding of the
hypotheses underlying the model, as well as a more transparent
interpretation of the parameters as individual quantities. Moreover, in
our model, part of this stochasticity will remain in the limit (through a
conditioning term) which could have been easily missed if
this convergence step was not carried out. Similar laws of large numbers have
been obtained frequently in the probabilistic literature on population
models \citep{kurtz1981approximation, oelschlager1990limit,
fournier2004microscopic}, in particular in an epidemic context
\citep{clemenccon2008stochastic, barbour2013approximating,
britton2019stochastic, pang2021functional, foutel2022individual}.
Our model draws inspiration from recent works on similar
non-Markovian epidemics \citep{pang2021functional, forien2021epidemic,
foutel2022individual, duchamps2021general}, and the limiting equations we
obtain are connected to classical time-since-infection models in epidemiology
\citep{diekmann1977limiting, diekmann1995legacy}.

Lastly, we would like to acknowledge the work of \cite{forien22}, who
consider a model very similar to ours (but without explicit vaccination),
derive rigorously the scaling limit of the epidemic, and give criteria
for the existence and asymptotic stability of an endemic equilibrium. We
emphasize that, despite the striking similarities, the two models were
formulated independently and most of the results presented in our work
were obtained before that in \cite{forien22} were made available.
Moreover, the main aim of our work is to draw public health insights from
our model, letting sometimes mathematical rigour aside, whereas
that of \cite{forien22} is mathematically much more accomplished and is
targeted to an audience of probabilists. We believe that the two
approaches offer complementary perspectives on the problem of waning
immunity and endemicity.

The rest of this article is organized as follows. We start by describing   
our stochastic model and its large population size limit in
Section~\ref{S:methods}. Then, in Section~\ref{S:publichealth}, we study
the long-time behavior of the limiting equations and provide a simple
criterion for the existence of an endemic equilibrium. In the following two
sections, we examine the dependence of this criterion on the vaccination
parameters to draw some public health insights from our model. We start
with a general discussion in Section~\ref{S:impactVac}, and consider two
more specific applications in Section~\ref{S:twoGroups}, which require us
to make a straightforward extension of our model to multiple groups. The well-posedness of the main PDE, introduced in Section~\ref{S:methods}, is proved in Section~\ref{sec:PDE}. Finally, a discussion on the model, the hypotheses and results is provided in Section~\ref{sec:discussion}.

\section{The model}
\label{S:methods}

\subsection{Model description}
\label{SS:model}

\paragraph{The dynamics without vaccination.}
We consider the spread of a disease in a closed population of fixed size
$N$, started at some reference time $t=0$ at which the state of the
epidemic is known. Each individual in the population is characterized by
two random quantities that change through time: its infectiousness,
giving the rate at which it transmits the disease, and its
susceptibility, corresponding to the probability that it gets reinfected
upon contact with an infected individual. Individuals are labeled by
$i\in\{1, \dots, N\}$, and the infectiousness of individual $i$ at time
$t$ is denoted by $\lambda^N_i(t)$ while its susceptibility is denoted by
$\sigma^N_i(t)$. We start by describing the dynamics of the epidemic in the
absence of vaccination, and then indicate how vaccines are included to
the model.  The typical evolution of the susceptibility and
infectiousness of an individual is represented in Figure~\ref{fig:individualTraj}.

Consider a focal individual $i$. In the absence of vaccination, it goes
repeatedly through two states. An infectious state (denoted by $I$),
where it cannot be infected but can spread the disease ($\lambda^N_i(t)
\ge 0$ and $\sigma^N_i(t) = 0$) and a susceptible state (denoted by $S$),
where it does not spread the disease but can be reinfected
($\lambda^N_i(t) = 0$ and $\sigma^N_i(t) \ge 0$). The transition
from $S$ to $I$ corresponds to the individual being infected and that 
from $I$ to $S$ to it recovering from the disease. Recovering might
confer partial or even full immunity. As an $S$ individual might be
partially immune (if $\sigma^N_i(t) < 1$), it is important to note that
our definition of a susceptible individual is different from the usual
one in differential equation models.


Upon infection, say at time $\tau$, individual $i$ enters the
$I$ state and samples a random function $\lambda \colon [0, T_I) \to [0,
\infty)$ according to a given distribution $\mathcal{L}_\lambda$. The
length $T_I$ of the domain of $\lambda$ is random, and we think of it as
being part of the definition of $\lambda$. The individual
remains in the $I$ state for a time period of length $T_I$, after
which it \emph{recovers} from the disease and moves to state $S$.
During its infectious period, it cannot get reinfected, and its
infectiousness is described by the function $\lambda$, that is, 
\[
    \forall a < T_I,\quad \lambda^N_i(\tau+a) = \lambda(a), 
    \quad \sigma^N_i(\tau+a) = 0.
\]

Once an individual has recovered from the disease, here at time
$\tau'=\tau+T_I$, it cannot spread the disease anymore and acquires an
immunity against reinfection that wanes. It enters the $S$ state. To
model that immunity is waning, the focal individual samples 
an independent random susceptibility $\sigma \colon [0, \infty) \to [0,
1]$ according to another given distribution $\mathcal{L}_\sigma$ on the
set of non-decreasing functions. We define
\[
    \forall a \ge 0,\quad \lambda^N_i(\tau'+a) = 0,
    \quad \sigma^N_i(\tau'+a) = \sigma(a).
\]
The susceptibility gives the probability to be infected upon exposure to
the disease. More precisely, individual $i$ gets reinfected at time $t$
at rate $\sigma^N_i(t)\Lambda^N(t)$  with $\Lambda^N$ the force of infection of the disease defined by
\begin{equation}\label{eq:infect-rate}
    \forall t \ge 0,\quad \Lambda^N(t) \coloneqq \frac{1}{N} \sum_{i=1}^N
    \lambda^N_i(t).
\end{equation}
When such an event occurs, individual $i$ goes back to the $I$ state
and we reproduce the above two steps independently. Note that since
$\sigma^N_i(t) = 0$ when $i$ is in state $I$, only susceptible individuals
can be reinfected.
The interpretation of the latter expression is that each infectious
individual, at a rate proportional to its infectiousness, makes a contact
targeted to an individual chosen uniformly in the population. This
contact leads to an infection with a probability given by the
susceptibility of the target individual. 

\begin{figure}
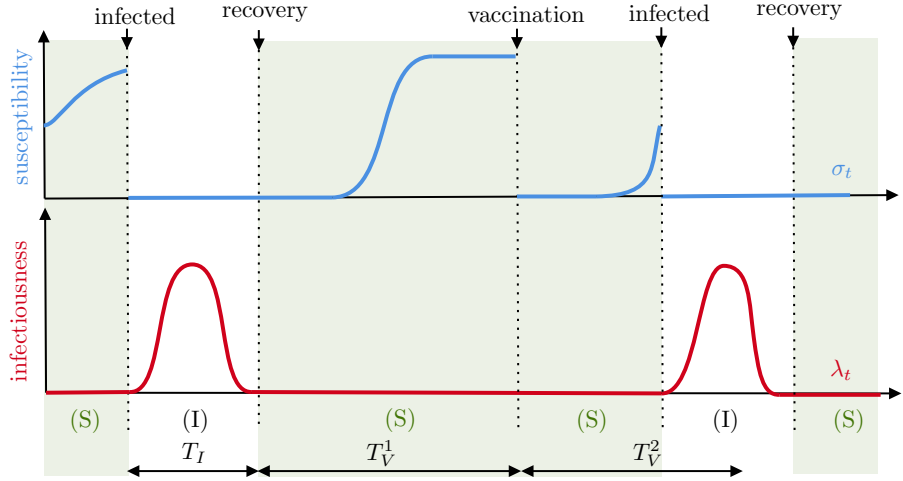

\centering
\include{tik-fig1}
\caption{Typical evolution of the
    susceptibility $\sigma$ (in \textcolor{couleurbleu}{blue}, on top) and  infectiousness $\lambda$ (in \textcolor{couleurrouge}{red}, below) of an individual.}
    \label{fig:individualTraj}
\end{figure}

\paragraph{Recovering the SIR, SIS, and SIRS models.} Let us rapidly illustrate how to
recover the usual compartmental models from the previous definitions.
Suppose that $\lambda$ and $\sigma$ are given by
\[
    \forall a < T_I,\quad \lambda(a) = \beta,
    \qquad 
    \forall a \ge 0,\quad \sigma(a) = \indic_{\{ a \ge T_R\}}
\]
where $T_I$ is exponentially distributed with parameter $\gamma$, and
$T_R$ is a random variable giving the length of the immune period. Then,
infectious individuals yield new infections at a constant rate $\beta$,
and stop being infectious (that is, recover) at rate $\gamma$. After
having recovered, an individual is completely immune to reinfections
($\sigma(a) = 0$) for a random duration $T_R$, after which it becomes
fully susceptible ($\sigma(a) = 1$). Depending on the distribution of
$T_R$, we can recover either the SIR, SIS or SIRS model. If $T_R =
\infty$ a.s., individuals are permanently immune following an infection
and our model becomes a stochastic version of the SIR model. If $T_R = 0$
a.s., individuals build no immunity following an infection, and this
corresponds to a stochastic SIS model. Finally, a stochastic SIRS model
can be obtained by letting $T_R$ be exponentially distributed. Obviously,
our model can account for much more complex situations, where
infectiousness varies during the infectious period, and immunity is
gradually lost following infection.

\paragraph{Adding vaccines.}
The previous rules describe the dynamics of the epidemic in the absence
of vaccination. We model vaccines by assuming that vaccination has the
same effect as (natural) immunization by the disease:
upon vaccination the susceptibility of an individual is ``reset'' to a
new independent random curve $\sigma'$ with law $\mathcal{L}_\sigma$, as
if it had entered the $S$ state after having recovered from the disease.
We further assume that infectious individuals are not vaccinated, and
that susceptible individuals are vaccinated recurrently according to a
renewal process until they are reinfected. In particular, it might occur
that an individual fully immune to the disease gets vaccinated.

More formally, if $\tau''$ denotes a (random) time at which $i$ is vaccinated, we
set
\[
    \forall a \ge 0,\quad \lambda^N_i(\tau''+a) = 0,
    \quad \sigma^N_i(\tau''+a) = \sigma'(a),
\]
for an independent random variable $\sigma'$ with distribution
$\mathcal{L}_\sigma$. After this event, a random independent
duration $T_V$ is sampled according to some law $\mathcal{L}_V$, which
gives the waiting time until individual $i$ receives its next vaccine
dose. If $i$ has not been reinfected by time $\tau'' + T_V$, we reiterate
the above two steps (resampling the susceptibility and a future
vaccination time) independently. This process goes on until the
individual is reinfected and goes back to the $I$ state.

The dynamics of the population can be obtained by carrying out the previous
steps altogether for each of the $N$ individuals. Every time an
individual is affected by an event (recovery, infection, vaccination), it
samples its new susceptibility or infectiousness independently of all
other individuals and of the past dynamics, according to $\mathcal{L}_\lambda$ 
in case of an infection or $\mathcal{L}_\sigma$ otherwise.

\subsection{Mathematical construction of the model}
\label{SS:formalModel}

We now present a formal description of the model.

\paragraph{Initial condition.} 
We suppose that the epidemic has been spreading for a long enough time
that all individuals have been infected or vaccinated at least once at
$t=0$. The initial condition we consider could easily be modified to
encompass more general scenarios. We assume that a fraction $I_0$ of
individuals are infected at $t=0$, and that $I_0 \in (0,1)$.
We assign to
each individual an initial state ($I$ or $S$), age, susceptibility, and
infectiousness independently in the following way. Consider a focal
individual $i \in \{1,\dots, N\}$. 

\begin{itemize}
\item We record the initial state of individual $i$ as a random variable $C_{i,0} \in \{S, I\}$ such
that 
\[
    \P(C_{i,0} = I) = 1 - \P(C_{i,0} = S) = I_0.
\]
\item Individuals are assigned an initial age $A_i(0)=-\tau_{i,0}$ such that the age is distributed according to a probability density $h_I$ for
an $I$ individual and to a probability density $h_S$ for an $S$ individual. 
\begin{itemize}
    \item[{\footnotesize$\circ$}] If individual $i$ is in the $I$ state:
        Conditional on $A_i(0)$, it is assigned an initial infectiousness
        $(\lambda_{i,0},T_{I,i,0})$ distributed as $\mathcal{L}_\lambda$ 
        conditional on $T_I > -\tau_{i,0}$, so that the individual
        remains infectious at time $t=0$. Before time $T_{I,i,0}$, the
        age of $i$ is $A^N_i(t) =
        A_i(0)+t$ and its infectiousness is
        \[
            \forall t < T_{I,i,0},\quad \lambda^N_i(t) = \lambda_{i,0}(t+A_i(0)),
            \quad \sigma^N_i(t) = 0.
        \]
        After time $T_{I,i,0}$, individual $i$ enters the $S$ state and
        follows the dynamics described above.
    \item[{\footnotesize$\circ$}] If individual $i$ is in the $S$ state:
        Conditional on $A_i(0)$, it is assigned two independent
        variables: a susceptibility $\sigma_{i,0}$ distributed as
        $\mathcal{L}_V$ and an independent initial vaccination time
        $T_{V,i,0}$, with law $\mathcal{L}_V$ conditional on $T_V >
        A_i(0)$. Again, the age and susceptibility of $i$ until time
        $T_{V,i,0}$ or until it gets infected are $A^N_i(t) = A_i(0) + t$
        and
        \[
            \sigma^N_i(t) = \sigma_{i,0}(t+A_i(0)).
        \]
\end{itemize}
\end{itemize}
All these variables are assigned independently for different individuals.

\paragraph{Spread of the epidemic.} Consider, for each $i$, three
independent i.i.d.\ sequences, $(\lambda_{i,k}, T_{I,i,k};\, k \ge 1)$,
$(\sigma_{i,k};\, k \ge 1)$, $(T_{V,i,k};\, k \ge 1)$ distributed as
$\mathcal{L}_{\lambda}$, $\mathcal{L}_{\sigma}$ and $\mathcal{L}_V$. We
also introduce for each $i$ an auxiliary sequence of independent
exponential random variables $(E_{i,k}; k \geq 1)$ with unit mean,
independent of the previous sequences.

From these random variables and the initial condition, we will construct
for each $i$ two sequences of random variables $(\tau_{i,k};\, k \ge 1)$ and 
$(C_{i,k};\, k \ge 1)$ that represent respectively the time at which $i$
experiences its $k$-th event (infection, recovery, or vaccination), and
its state after this $k$-th event ($I$ or $S$). Assuming that these
variables are constructed, the age $A_i(t)$, the state $C_i(t)$, the
infectiousness $\lambda_i^N(t)$ and the susceptibility $\sigma_i^N(t)$ of
individual $i$ at time $t$ are simply given by 
\begin{align*}
    A_i(t)&=t-\tau_{i,K_i(t)},\quad
    C_i(t)=C_{i,K_i(t)}, \nonumber\\
    \lambda_i^N(t) &= \lambda_{i,K_i(t)}\big(A_i(t)\big) 
    \indic_{\{C_i(t)=I\}},  \nonumber\\
    \sigma_i^N(t) &= \sigma_{i,K_i(t)}\big(A_i(t)\big)
    \indic_{\{C_i(t)=S\}},
\end{align*}
where 
\[
    \forall t \ge 0,\quad K_i(t) = \sup \{ k \ge 0 : \tau_{i,k} < t \}
\]
is the number of events experienced by $i$ at time $t$. We recall that
the force of infection at time $t$ is defined as 
\[
    \Lambda^N(t) = \frac{1}{N} \sum_{i=1}^N \lambda_i^N(t),
\]
with the convention that $\Lambda^N \equiv 0$ for negative times.

Let us now construct $(\tau_{i,k};\, k \ge 1)$ and $(C_{i,k};\, k \ge 1)$
inductively. Suppose $\tau_{i,k}$ and $C_{i,k}$ have been constructed. We
distinguish between two cases. If $C_{i,k} = I$, individual $i$
eventually recovers so that we set $C_{i,k+1} = S$. This recovery occurs
after a period of length $T_{I,i,k}$, and we define $\tau_{i,k+1} = \tau_{i,k} +
T_{I,i,k}$. If $C_{i,k} = S$, the next event experienced by individual
$i$ is either a vaccination or a reinfection. We use $E_{i,k}$
to define the time of reinfection in the absence of vaccination as
\begin{equation} \label{eq:reinfectionTime}
    Z_{i,k} = \inf \Big\{ a \ge 0 : \int_0^a
    \Lambda^N(\tau_{i,k}+u)\sigma_{i,k}(u) \diff u > E_{i,k} \Big\}.
\end{equation}
Note that $Z_{i,k}$ corresponds to the first atom of a Poisson point process  
with random intensity $\Lambda^N(\tau_{i,k} + \cdot) \sigma_{i,k}(\cdot)$.
If $T_{V,i,k} > Z_{i,k}$, the individual gets reinfected before it is
vaccinated. We set $C_{i,k+1} = I$ and $\tau_{i,k+1} = \tau_{i,k} +
Z_{i,k}$. Otherwise, the individual is vaccinated before being
reinfected, and we set $C_{i,k+1} = S$, and $\tau_{i,k+1} = \tau_{i,k} +
T_{V,i,k}$.

Overall, apart from the initial condition, the dynamics of the epidemic
depends on four parameters: the population size $N$, the distribution of
the infectiousness curve $\mathcal{L}_\lambda$, the distribution of the
susceptibility curve $\mathcal{L}_\sigma$, and the distribution of the
duration between two vaccinations $\mathcal{L}_V$. We will always assume
that $T_I$ and $T_V$ have a density and a finite expectation.
Some realizations of the model are displayed in Figure~\ref{fig:simulationModel}.

\begin{figure}
    \centering
    \includegraphics[width=\textwidth]{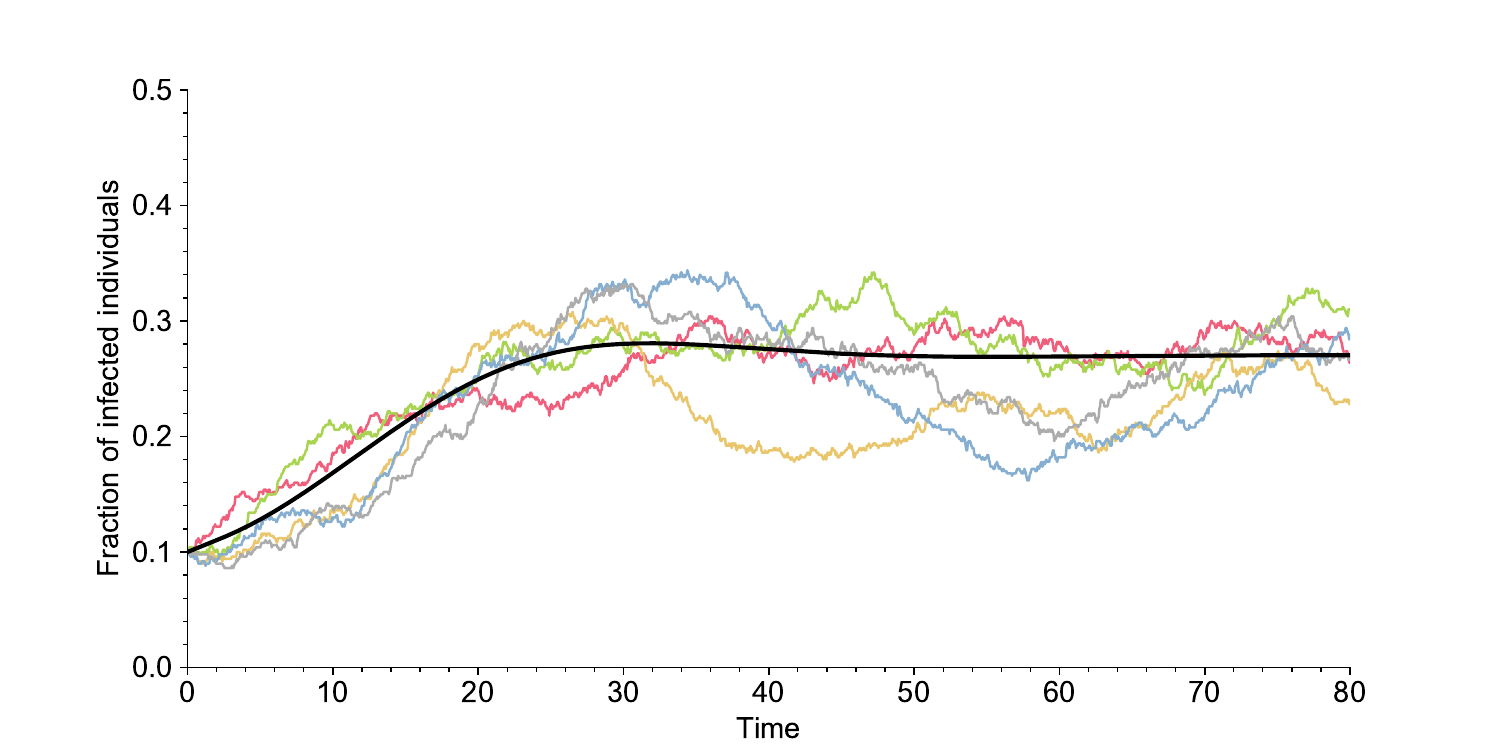}
    \caption{Independent simulations of the model (colored lines) for $N = 500$ 
        and of its deterministic limit, as $N\to\infty$ (solid {\bf black line}). All
            parameter values are given in Table~\ref{tab:parametersValue} (Appendix~\ref{sec:numerical:simulations}). The simulations are
    initialized with a fraction $I_0 = 0.1$ of infectious individuals.}
    \label{fig:simulationModel}
\end{figure}

\paragraph{Population age structure.}
Since the infectiousness and susceptibility are varying with time in our
model, the state of the epidemic is not accurately described by counting
the number of infectious and susceptible individuals. Deriving the large
population size limit of the model requires to record for each individual
a duration, which is the time elapsed since the last event (infection, recovery,
or vaccination) that they experienced. We refer to this duration as the
class age, or simply the age. Thus, the age of an $I$ individual is
the time elapsed since its infection, which is the classical definition of
the age-of-infection in epidemic models \citep{thieme1993may,
inaba2017age}. For an $S$ individual it is the time since its last
vaccination or recovery event, where we recall that an individual
is said to recover when it moves from state $I$ to state $S$.
Recall that $A^N_i(t)$ is the age of individual $i$ at
time $t$, and $C^N_i(t) \in \{I, S\}$ its state. The
relevant quantity that we will study is the age and state structure of
the population, which we encode as the following two point measures on
$[0, \infty)$
\begin{equation} \label{eq:ageStructure}
    \forall t \ge 0,\quad \nu_I^N(t) = \frac{1}{N} \sum_{\substack{i = 1,\dots,N\\C_{i}(t)=I}}
    \delta_{A^N_i(t)}, \quad
    \nu_S^N(t) = \frac{1}{N} \sum_{\substack{i = 1,\dots,N\\C_{i}(t)=S}}\delta_{A^N_i(t)}.
\end{equation}
The measure $\nu^N_I(t)$ (resp.\ $\nu^N_S(t)$) has one atom for each
infectious (resp.\ susceptible) individual at time $t$, with mass $1/N$
and whose location is the age of that individual.

\subsection{Large population size limit}
\label{SS:MKV}

The above stochastic model is too complicated to be studied directly.
Instead, we will identify its large population size limit and use it
as an approximation of the finite population model to study the
efficiency of a vaccination strategy. 
We rely on a classical approach in statistical physics called
\emph{propagation of chaos} (sometimes also referred to as
\emph{molecular chaos}). The general idea is as follows: consider a
stochastic system of $N$ particles, initially independently
distributed, with mean-field interaction between them. Focusing on a
fixed number $k$ of particles, under general assumptions the mean-field
interaction averages out as the size of the system $N$ tends to $\infty$,
and the $k$ particles behave independently in the limit. (The terminology propagation
of chaos refers to the propagation in time of the initial independence
of the particles when $N \to \infty$.) If the system is exchangeable,
this asymptotic independence entails a law of large numbers for the
empirical measure of the system. The limiting distribution usually has a
density characterized as the unique solution of a nonlinear PDE.
We refer the reader to \cite{sznitman89, Meleard1996} for a detailed
description of these concepts in the original context of kinetic theory, as
well as to the recent exhaustive review by \cite{ChaintronDiez22-I,
ChaintronDiez22-II}. For applications of these ideas in biological
contexts, we refer to \cite{chevallier17, Locherbach-Fournier16} for
models of interacting neurons, and finally to \cite{forien22} for a
rigorous approach on a similar epidemiological model to ours.

We now apply these ideas to our model. The $N$ particles correspond to
the $N$ individuals in the population, which interact in a mean-field
way through infections. Since this approach is rather standard in the
probabilistic literature, we only outline the usual arguments here.
Recall the definition \eqref{eq:ageStructure} of the empirical age
distributions in the two compartments at time $t$, $\nu_I^N(t)$ and
$\nu_S^N(t)$. We assume in this section that $\nu_I^N(t)$ (resp.\
$\nu_S^N(t)$) converges (in distribution as random measures) as $N \to
\infty$ to a deterministic limiting measure with density $(I(t,a);\, a
\ge 0)$ (resp.\ $(S(t,a);\, a \ge 0)$). We will identify the equations
fulfilled by these limits. For any continuous bounded test function $f$
we have 
\begin{align*}
    \int_0^\infty I(t,a)f(a)\diff a 
    &= \lim_{N \to \infty} \E[\angle{\nu_I^N(t), f}] \\
    &= \lim_{N \to \infty} \frac{1}{N} \sum_{i=1}^N \E[f(A_i^N(t)) \indic_{\{C_i^N(t)=I\}} ] \\
    &= \lim_{N \to \infty} \E[f(A_1^N(t)) \indic_{\{C_i^N(t)=I\}} ],
\end{align*}
where the last line follows from the exchangeability of the system
and where we have used the standard notation $\angle{\mu, f} = \int f(x)
\mu(\diff x)$. A
similar computation holds for $\nu^N_S(t)$. This shows that the limiting age 
structures of $I$ and $S$ individuals correspond to the limit in
distribution of the age of one typical individual in the population, on
the event that it is in state $I$ or $S$ respectively. Thus, we only need
to understand the dynamics of a single individual, say individual $i=1$,
in the limit $N \to \infty$.

According to the rules of the model, individual $1$ only depends on the
other individuals in the population through its infection rate
$\sigma^N_1(t) \Lambda^N(t)$, with $\Lambda^N$ given by
\eqref{eq:infect-rate}. Assuming that $\Lambda^N(t) \to \Lambda(t)$ as $N
\to \infty$ for some deterministic function $\Lambda \coloneqq
(\Lambda(t);\, t \ge 0)$, the limit in distribution of the age, state,
infectiousness, and susceptibility of individual $1$ is simply obtained
by replacing $\Lambda^N$ with its limit $\Lambda$ in the model
description of Section~\ref{SS:model}. Using a similar notation as in
Section~\ref{SS:model}, let us denote by $(\lambda^\Lambda(t),
\sigma^\Lambda(t), A^\Lambda(t), C^\Lambda(t))$ the limit of the
infectiouness, susceptibility, age and state of individual $i=1$ at time
$t$, as $N \to \infty$. (That is, when $\Lambda^N$ is replaced by
$\Lambda$.) Now, by exchangeability as above
\begin{align*}
    \Lambda(t)
    &= \lim_{N \to \infty} \E[\Lambda^N(t)] 
    = \lim_{N \to \infty} \frac{1}{N} \sum_{i=1}^N \E[\lambda^N_i(t)]\\
    &=  \lim_{N \to \infty}  \E[\lambda^N_1(t)]
    = \E[ \lambda^\Lambda(t) ].
\end{align*}
This puts the consistency constraint on $\Lambda$ that 
\begin{equation} \label{eq:McKeanVlasov}
    \forall t \ge 0,\quad \Lambda(t) = \E[ \lambda^\Lambda(t) ]
\end{equation}
should hold. In the terminology of propagation of chaos, a stochastic
system satisfying \eqref{eq:McKeanVlasov} is called a solution to a
McKean--Vlasov equation, or also a solution to a non-linear equation. It
can be shown (see Proposition~\ref{prop:uniqueness-sol-PDE} in
Section~\ref{sec:proof-wellposedness})  that for our stochastic model, there
exists a unique solution to the McKean--Vlasov equation
\eqref{eq:McKeanVlasov}. We denote it by $\Lambda^*$, and by
$(\lambda^*(t), \sigma^*(t),A^*(t), C^*(t))$ the corresponding quantities. The
following result, that we state without proof, identifies the limit of
the age and class structure of our model to the distribution of $(A^*(t),
C^*(t))$ of the solution to the McKean--Vlasov equation. It has been
proved in \cite{forien22} by making the above heuristic arguments rigorous. 

\begin{theorem}[\cite{forien22}, Theorem~3.2] \label{thm:lawLargeNumbers}
    Suppose that there exists $\lambda_{\max}$ such that $\lambda(a) \le
    \lambda_{\max}$ almost surely for all $a \ge 0$. Then for any $t \ge 0$ we
    have 
    \[
        \lim_{N \to \infty} \nu^N_I(t) = I(t,a) \diff a,\qquad
        \lim_{N \to \infty} \nu^N_S(t) = S(t,a) \diff a
    \]
    in distribution for the topology of weak convergence. Furthermore,
    $(I(t,a);\, a \ge 0)$ is the density of $A^*(t)$ on the event
    $\{C^*(t) = I \}$, and $(S(t,a);\, a \ge 0)$ that on the event
    $\{C^*(t) = S \}$, where $(\lambda^*(t), \sigma^*(t), A^*(t),
    C^*(t);\, t \ge 0)$ is the (unique) solution to the above
    McKean--Vlasov equation \eqref{eq:McKeanVlasov}.
\end{theorem}

\subsection{PDE formulation of the limit}
\label{SS:largePopSize}

The previous section has characterized the law of large numbers limit of
the age structure of the epidemic in terms of the distribution of a
stochastic system representing the limiting dynamics of a single
individual in the population (of infinite size). We now provide a PDE formulation for this
distribution, which can be thought of as the forward Kolmogorov equation
associated to the previous stochastic process, although note that it is
not a Markov process.

\paragraph{Description of the limit.}
Consider the following PDE, whose terms will be introduced throughout
this section, 
\begin{align} 
    \begin{split} \label{eq:main}
    \partial_t I(t,a) + \partial_a I(t,a) &= - \mu_I(a) I(t,a) \\
    \partial_t S(t,a) + \partial_a S(t,a) &= - \mu_V(a) S(t,a) - \Lambda(t)
    \E_{t,a}\big[ \sigma(a) \big] S(t,a) \\
    I(t,0) &= \Lambda(t) \int_0^\infty \E_{t,a}\big[ \sigma(a) \big]
        S(t,a) \diff a\\
    S(t,0) &= \int_0^\infty \mu_I(a) I(t,a) \diff a + \int_0^\infty \mu_V(a) S(t,a) \diff a
    \end{split}
\end{align}
with initial conditions
\begin{align*}
    I(0, a) &= I_0 h_I(a) \\
    S(0, a) &= (1-I_0) h_S(a),
\end{align*}
where $I_0\in(0,1)$ is the fraction   of
infected individuals at $t=0$ and $h_I$ (resp. $h_S$) is age density of initially
infectious (resp. susceptible)  individuals. 

Both the equation for $S$ and $I$ have a transport term corresponding to
the aging phenomenon, and some removal terms corresponding to infections,
vaccinations, and recoveries. Recovery (resp.\ vaccination) occurs at
rate $\mu_I(a)$ (resp.\ $\mu_V(a)$) at age $a$, where $\mu_I(a)$ and
$\mu_V(a)$ are the (age-dependent) recovery and vaccination rates
respectively, which we assume to exist:
\begin{equation} \label{eq:hazardFunction}
    \forall a \ge 0,\quad \mu_I(a)\diff a = \frac{\P\big( T_I \in [a,a+\diff
    a]\big)}{\P( T_I \ge a)}, \quad
    \mu_V(a) \diff a = \frac{\P\big( T_V \in [a,a+\diff
    a]\big)}{\P( T_V \ge a)}.
\end{equation}
Newly recovered and vaccinated individuals become susceptible with age
$a=0$ (typically implying being immune), yielding the two
integrals in the age boundary condition for $S$.

The last and most interesting term corresponds to new infections. An
individual is infected at a rate which is the product of its own
susceptibility and of the force of infection in the population. In the
limiting system, the force of infection is obtained by integrating the
age-dependent infectiousness of $I$ individuals over the age structure:
\begin{equation} \label{eq:Lambda}
    \forall t \ge 0,\quad \Lambda(t) = \int_0^\infty
    \E\big[\lambda(a) \mid T_I > a\big] I(t,a) \diff a.
\end{equation}
We set $\Lambda \equiv 0$ for negative times. This is the usual
expression for the force of infection in an epidemic model structured by
time-since-infection \citep{kermack1927contribution, diekmann1995legacy,
brauer2005kermack}. We define
\begin{equation} \label{eq:biasExpectation}
    \E_{t,a}\big[\sigma(a) \big] = \E\Big[ \sigma(a) e^{-\int_0^a
        \Lambda(t-a+u) \sigma(u) \diff u} \Big] 
        \;\Big/\; 
        \E\Big[ e^{-\int_0^a \Lambda(t-a+u) \sigma(u) \diff u} \Big] 
\end{equation}
to be the expected susceptibility of an $S$ individual with age $a$ at
time $t$. The exponential term reflects that a susceptible individual
with age $a$ at time $t$ is conditioned on not being infected between
$t-a$ and $t$, which biases $\sigma$ in favor of a low susceptibility
during this time period. This is an interesting example where the
stochasticity of the underlying individual-based model does not entirely
vanish in the large population size limit. Disregarding this
stochasticity changes the limiting equations and hence the prediction of the model,
even at the macroscopic scale. A similar conditioning is considered in
\cite{breda12}. Note that if $\sigma$ is deterministic the bias vanishes,
that is, $\E_{t,a}[\sigma(a)] = \sigma(a)$. Our set of equations then becomes a
version of the reinfection model of Kermack and McKendrick \citep{kermack32,
kermack33, inaba01} in a closed population which incorporates
vaccination. In our model, this amounts to discarding the
inter-individuals variation in the immunity waning.

Finally we introduce the basic reproduction number $R_0$ as
\begin{equation}\label{eq:R0}
    R_0 = \int_0^{\infty}\E[\lambda(a)]\mathrm{d}a,
\end{equation}
which we assume to be finite. As usual, $R_0$ represents the average
number of secondary cases generated by an infected individual in a fully
susceptible population.

\paragraph{Weak solution and well-posedness.}
We now introduce the definition of a weak solution to a general
transport equations. Let $F$ be a locally integrable function on
$\mathbb{R}_+\times\mathbb{R}_+$. We say that $(f(t,a);\, t,a \ge 0)$ is
a weak solution to 
\begin{equation} \label{eq:strongSolution}
    \partial_t f(t,a) + \partial_a f(t,a) = F(t,a) f(t,a)
\end{equation}
if 
\begin{gather*}
    \forall a \le t,\quad f(t,a) = f(t-a, 0) \exp\Big( \int_0^a F(t-a+u, u) \diff u\Big) \\
    \forall a \ge t,\quad f(t,a) = f(0, a-t) \exp\Big( \int_{a-t}^a F(t-a+u, u) \diff u\Big).
\end{gather*}
This definition is motivated by a formal application of the method of
characteristics. Suppose that $f$ is a strong solution to the previous
equation (in the sense that $f$ is continuously differentiable and its
partial derivatives verify \eqref{eq:strongSolution} in the interior of
the domain). We see that, along the characteristic line such $t-a$ is
constant, $f$ solves a first order linear differential equation. More
precisely, by differentiating the map $g \colon u \mapsto f(t-u, a-u)$,
it solves
\[
    g'(u) = -F(t-u, a-u) g(u).
\]
Solving this equation on $(0, t \wedge a)$ and noting that 
$g(0) = f(t,a)$ lead to the above expression.
With this notion of weak solution, we can show that equation
\eqref{eq:main} is well-posed under mild technical assumptions.

\begin{proposition}\label{prop:uniqueness-sol-PDE}
Equation~\eqref{eq:main} has a unique weak solution on the  Skorokhod
space $\mathbb{D}(\R^+,\R^+)$, when the following conditions hold
\begin{itemize}
    \item there exists $\lambda_{\max}>0$ such that, $\forall a\geq 0$,   
        $\E[\lambda(a)]\leq \lambda_{\max}$,
    \item the density distribution functions of $T_I$ and $T_V$ and the functions
        \begin{equation} \label{eq:boundedIC}
        t\mapsto \int_0^\infty \mu_I(t+a)\mathrm{e}^{-\int_a^{a+t}\mu_I(u)\diff u}h_I(a)\diff a
        +\int_0^\infty \mu_V(t+a)\mathrm{e}^{-\int_a^{a+t}\mu_V(u)\diff u}h_S(a)\diff a
    \end{equation}
    are bounded.
\end{itemize}
\end{proposition}

The above result is proved in Section~\ref{sec:proof-wellposedness}.
The next result connects the PDE \eqref{eq:main} to the distribution of
the solution to the McKean--Vlasov \eqref{eq:McKeanVlasov}. It follows
from elementary manipulations of point processes, and we postpone its
proof until Section~\ref{sec:proofPDE}.
    
\begin{proposition} \label{prop:PDErepresentation}
    Let $(\lambda^*(t), \sigma^*(t),A^*(t), C^*(t))$ be the solution to
    the McKean--Vlasov equation~\eqref{eq:McKeanVlasov}. Then, if $I(t,
    \cdot)$ (resp.\ $S(t, \cdot)$) is the density of $A^*(t)$ on the
    event $C^*(t) = I$ (resp.\ $C^*(t) = S$), $(I(t,a);\, t,a \ge 0)$ and
    $(S(t,a);\, t,a \ge 0)$ are the unique weak solutions to
    \eqref{eq:main}.
\end{proposition}

\section{Long-term behavior of the epidemic}\label{S:publichealth}

\subsection{Equilibrium analysis}
\label{S:equilibria}

We are interested in the long-time behavior of equation \eqref{eq:main}.
If this PDE converges to an equilibrium, the equilibrium should be a
stationary solution of \eqref{eq:main}, that is, a solution which is
independent of $t$ and thus of the form 
\[
    \forall t,a \ge 0,\quad I(t,a) = I(a),\quad S(t,a) = S(a).
\]
Similarly, let $\Lambda$ be the quantity defined in \eqref{eq:Lambda},
but using the stationary age profile $(I(a);\, a \ge 0)$, and
$\E_a[\sigma(a)]$ be defined through \eqref{eq:biasExpectation}, using
the stationary force of infection $\Lambda$. Note that these quantities
no longer depend on the time variable. As is usual in similar epidemic
models, we distinguish between two types of equilibria: disease-free
equilibria, where there are no infected individuals in the population;
and endemic equilibria, where the disease persists in the population.

\paragraph{Disease-free equilibrium.}
First, suppose that $I \equiv 0$. Then the PDE reduces to a 
first order linear differential equation,
\[
    \forall a \ge 0,\quad S'(a) = -\mu_V(a) S(a),
\]
whose unique solution is 
\begin{equation} \label{eq:diseaseFree}
    \forall a \ge 0,\quad S(a) = S(0) \exp\Big( -\int_0^a \mu_V(u) \diff u \Big),
\end{equation}
where $S(0) = \E[T_V]^{-1}$ is so that $(S(a);\, a \ge 0)$ is a probability distribution.
This shows that \eqref{eq:main} always admits a unique disease-free equilibrium.
Note that this equilibrium could have been easily anticipated. In the
absence of infections, individuals only get vaccinated according to a
renewal process with renewal time distribution $T_V$. Equation
\eqref{eq:diseaseFree} is the stationary distribution of the time since
the last vaccination event for this renewal process.

\paragraph{Endemic equilibrium.}
We now turn our attention to endemic equilibria ($I\not\equiv 0$). Let us make some
computation to find an appropriate candidate. The two differential terms 
in \eqref{eq:main} are reduced to the following linear differential equations 
for $I$ and $S$:
\begin{gather*}
    \forall a \ge 0,\quad I'(a) = -\mu_I(a) I(a) \\
    \forall a \ge 0,\quad S'(a) = -\mu_V(a) S(a) - \E_a[\sigma(a)] \Lambda S(a).
\end{gather*}

Therefore any endemic equilibrium should fulfill that 
\begin{equation} \label{eq:eqI}
    \forall a \ge 0,\quad I(a) = I(0) \exp\Big(- \int_0^a \mu_I(u) \diff u\Big)
\end{equation}
and 
\begin{align}
    \nonumber
    \forall a \ge 0,\quad S(a) &= S(0) \exp\left(- \int_0^a \mu_V(u) \diff
    u - \Lambda\int_0^a \E_u[\sigma(u)] \diff u\right) \\
    &= S(0) \exp\left(- \int_0^a \mu_V(u) \diff u  \right)
    \E\left[ e^{- \Lambda \int_0^a \sigma(u) \diff u } \right].
    \label{eq:eqS}
\end{align}
In the last line we have used that $g(a) = \E[ e^{- \Lambda \int_0^a
\sigma(u) \diff u } ]$ is easily seen to solve $g'(a) =-
\Lambda\E_a[\sigma(a)] g(a)$, where at equilibrium $\E_a[\sigma(a)]=\E[
\sigma(a) e^{-\Lambda\int_0^a \sigma(u) \diff u} ] \;\big/\; \E[
e^{-\Lambda\int_0^a  \sigma(u) \diff u} ] $.

Recalling the definition of the force of infection \eqref{eq:Lambda},
then using \eqref{eq:eqI} and the definition of $R_0$ in \eqref{eq:R0},
\[
    \Lambda = \int_0^\infty I(a) \E[ \lambda(a) \mid T_I > a ] \diff a 
    = I(0) \int_0^\infty \E[\lambda(a)] \diff a = I(0) R_0.
\]
Using the boundary condition for $I${, the fact that $\sigma$ and $T_V$ are independent}, and the definition of
$\E_a[\sigma(a)]$,  
\begin{align*}
    I(0) &= \Lambda \int_0^\infty \E_a[\sigma(a)] S(a)\diff a
    = S(0) \int_0^\infty e^{-\int_0^a \mu_V(u) \diff u}
    \E\Big[ \Lambda \sigma(a) e^{- \Lambda \int_0^a \sigma(u) \diff u } \Big] 
    \diff a \\
    &= S(0) \int_0^\infty 
    \E\Big[ \indic_{\{T_V > a\}} \Lambda \sigma(a) e^{- \Lambda \int_0^a \sigma(u) \diff u } \Big] 
    \diff a 
    = S(0) \E\Big[ 1 - e^{- \Lambda \int_0^{T_V} \sigma(u) \diff u } \Big].
\end{align*}
Therefore, an endemic equilibrium should verify that 
\begin{equation} \label{eq:initialS}
    S(0) = I(0) \:\Big/\: \E\Big[ 1 - e^{- R_0 I(0) \int_0^{T_V} \sigma(u) \diff u } \Big].
\end{equation}
Together, these computations lead to the following criterion for the
existence of an endemic equilibrium.

\begin{proposition} \label{prop:endemic}
    There exists an endemic equilibrium for each positive solution $x$ of the
    equation $F_\e(x) = R_0$ with 
    \begin{equation} \label{eq:endemicity}
        F_\e(x) \coloneqq x \E[T_I] + x \frac{\displaystyle 
            \E\Big[ \int_0^{T_V} e^{-x \int_0^a \sigma(u) \diff u} \diff a\Big]}%
        {\displaystyle \E\Big[1 - e^{-x \int_0^{T_V} \sigma(u) \diff u} \Big]}.
    \end{equation}
    For a given solution $x$, the corresponding equilibrium is so that
    $I(0) = x / R_0$, $S(0)$ is given by \eqref{eq:initialS}, and
    $(I(a);\, a \ge 0)$ and $(S(a);\, a \ge 0)$ by \eqref{eq:eqI} and
    \eqref{eq:eqS} respectively.
\end{proposition}

\begin{proof}
    Suppose that $(I(a);\, a \ge 0)$ and $(S(a);\, a \ge 0)$
    are a stationary solution of \eqref{eq:main}. Then, \eqref{eq:eqS}
    shows that 
    \begin{align*}
        \int_0^\infty S(a) \diff a 
        &= S(0) \int_0^\infty \P(T_V > a)
        \E\Big[ e^{- R_0 I(0) \int_0^a \sigma(u) \diff u } \Big] \diff a \\
        &= S(0) \E\Big[ \int_0^{T_V} e^{-R_0 I(0) \int_0^a \sigma(u)
        \diff u} \diff a \Big].
    \end{align*}
    Combining this to \eqref{eq:eqI} and \eqref{eq:initialS} yields
    \[
        \int_0^\infty I(a) \diff a + \int_0^\infty S(a) \diff a
        = I(0) \E[T_I] + I(0) \frac{\displaystyle 
        \E\Big[ \int_0^{T_V} e^{- R_0 I(0) \int_0^a \sigma(u) \diff u} \diff a\Big]}%
        {\displaystyle \E\Big[1 - e^{- R_0I(0) \int_0^{T_V} \sigma(u) \diff u} \Big]}
        = 1
    \]
    so that setting $x = I(0) R_0$ leads to a solution of
    \eqref{eq:endemicity}.
    
    Conversely, let $x$ be a solution $F_\e(x) = R_0$. Define
    $(I(a);\, a \ge 0)$ and $(S(a);\, a \ge 0)$ as in the statement of
    the result. The computation we have made already shows that both
    differential terms and the boundary condition for $I$ are fulfilled.
    It is straightforward to check that the boundary condition for $S$ is
    also fulfilled. All what remains to check is that 
    \[
        \int_0^\infty I(a) \diff a + \int_0^\infty S(a) \diff a = 1,
    \]
    which holds by making the same calculation as in the first part of
    the proof and using that $x$ solves $F_\e(x) = R_0$.
\end{proof}

\subsection{The endemic threshold}
\label{SS:endemicCriterion}

From the characterization of the existence of endemic equilibria in the
previous section, we see that there exists a threshold for $R_0$ under
which there can be no endemic equilibrium. Indeed, since $F_\e$ is
continuous, we see that 
\[
    \text{there exists an endemic equilibrium} 
    \iff 
    \text{$R_0 \ge \inf_{(\epsilon, \infty)} F_\e(x)$ for some $\epsilon > 0$}.
\]
We would like to obtain an explicit expression for this threshold and to
study the uniqueness of an endemic equilibrium if it exists. This
requires to study the variations of the function $F_\e$. Let us start
with two specific cases for which we can study the variations
analytically. The two cases covered by this result are broad enough for
many interesting applications, in particular choosing $\sigma$ to be of
the form \eqref{eq:sigmaSIRS} leads to a stochastic version of the SIRS
model, with general durations.

\begin{proposition} \label{prop:endemicFunction}
    Suppose that $\sigma$ is either deterministic or of the form
    \begin{equation} \label{eq:sigmaSIRS}
        \forall a \ge 0,\quad \sigma(a) = \indic_{\{a \ge T_R\}}
    \end{equation}
    for some random duration $T_R$ with $\E[T_R] < \infty$. Then $F_\e$
    is increasing and there exists a unique endemic equilibrium if and
    only if $R_0 \Sigma > 1$ with
    \begin{equation}  \label{eq:sigma} 
        \Sigma \coloneqq \frac{\E\big[\int_0^{T_V} \sigma(a) \diff a\big]}{\E[T_V]}.
    \end{equation}
\end{proposition}

\begin{proof}
Define
\[
    \forall a \ge 0,\quad \phi(a) = \int_0^a \sigma(u) \diff u.
\]
The endemic function $F_\e$ can be written as 
\begin{equation}\label{eq:defF}
    \forall x > 0,\quad F_\e(x) = x \E[T_I] + \dfrac{\E\Big[ \int_0^{T_V} e^{-x \phi(a)} \diff a\Big]} {\E\Big[ \int_0^{T_V} \sigma(a)e^{-x \phi(a)} \diff a\Big]}.     
\end{equation}
We now distinguish between the two cases of the proposition.

\medskip
\noindent
\textit{Step susceptibility.} 
If $\sigma$ is of the form $\sigma(a)=\indic_{\{a\geq T_R\}}$ for some
random recovery duration $T_R$, we easily observe that
\begin{gather*}
    \E\Big[ \int_0^{T_V} e^{-x \phi(a)} \diff a\Big] 
    = \E[T_V\wedge T_R]+\E\left[\indic_{\{T_V>T_R\}}\int^{T_V}_{T_R}e^{-x (a-T_R)}\diff a\right] \\
    \E\Big[ \int_0^{T_V} \sigma(a)e^{-x \phi(a)} \diff a\Big] 
    = \E\left[\indic_{T_V>T_R}\int^{T_R}_{T_R}e^{-x (a-T_R)}\diff a\right].
\end{gather*}
Consequently, $F_\e(x)=x\E[T_I]+1+\E[T_V\wedge T_R]/\E[\int_0^{(T_V-T_R)_+}e^{-xa}\diff a]$, 
which is obviously an increasing function.

\medskip
\noindent
\textit{Deterministic susceptibility.}
From \eqref{eq:defF}, we note that 
\[
    F_\e'(x) = \E[T_I] + \frac{h(x)} {\E\Big[ \int_0^{T_V} \sigma(a)e^{-x \phi(a)} \diff a\Big]^2}
\]
with
\begin{multline*}
    h(x) = \E\left[ \int_0^{T_V} e^{-x \phi(a)} \diff a\right] 
    \E\left[ \int_0^{T_V} \sigma(a)\phi(a)e^{-x \phi(a)} \diff a \right] \\
    - \E\left[ \int_0^{T_V} \phi(a)e^{-x \phi(a)} \diff a\right] 
    \E\left[ \int_0^{T_V} \sigma(a)e^{-x \phi(a)} \diff a\right].
\end{multline*}
Let $(T, \tilde{T})$ be a pair of independent copies of $T_V$.
Since $\phi$ is deterministic
\begin{align*}
    h(x) &=  \E\left[ \int_0^T\int_0^{\tilde{T}} 
    \sigma(b)\phi(b)e^{-x (\phi(a) + \phi(b))} \diff a \diff b\right] 
    - \E\left[ \int_0^T \int_0^{\tilde{T}} \phi(a) \sigma(b)e^{-x (\phi(a)+\phi(b))} 
    \diff a \diff b \right] \\
    &=\E\left[\int _0^T\int_0^{\tilde T} \left(\sigma(a)-\sigma(b)\right)\phi(a)e^{-x( \phi(a)+ \phi(b))}   \diff a\diff b\right]\\
    &=\frac{1}{2}\E\left[\int _0^T\int_0^{\tilde T}
    \left(\sigma(a)-\sigma(b)\right)\left(\phi(a)-\phi(b)\right)e^{-x( \phi(a)+ \phi(b))}   \diff a\diff b\right] \ge 0,
\end{align*}
where we conclude using that $(\sigma(a)-\sigma(b))(\phi(a)-\phi(b)) \ge 0$ 
because both functions are non-decreasing. This shows that $F_\e'(x) > 0$,
proving that $F_\e$ is increasing.
\end{proof}

The critical value $1/\Sigma$ in \eqref{eq:sigma} corresponds to the limit of
$F_\e$ at $0$, which we can always compute (without any assumption on
$\sigma$) as 
\[
    \frac{1}{\Sigma} 
    = \lim_{x \to 0} F_\e(x) = \frac{\E[T_V]}{\E\big[\int_0^{T_V} \sigma(a) \diff a\big]}.
\]
As a consequence, it is easily seen that in general there exists at least
one endemic equilibrium if $R_0 \Sigma > 1$. However, having uniqueness
of this equilibrium and absence of endemic equilibrium when $R_0\Sigma \le 1$
requires that $F_\e$ is increasing, which we were not able to prove in
general. Though, numerical simulations of $F_\e$ suggest that it is an
increasing function for a larger class of nondecreasing random curves $\sigma$, 
and we expect this to hold more generally, see
Figure~\ref{fig:variationEndFunction} in Appendix~\ref{sec:model:param}. 

The criterion $R_0\Sigma > 1$ has an interesting interpretation in terms of the
survival of a branching process \citep{athreya1971branching}.
Writing
\begin{equation*} 
    \Sigma = 
    \frac{\E\big[ \int_0^{T_V} \sigma(a) \diff a \big]}{\E[T_V]}
    = \int_0^\infty \E[\sigma(a)] S(a) \diff a,
\end{equation*}
we note that $\Sigma$ corresponds to the mean susceptibility of the
population at the disease-free equilibrium given by \eqref{eq:diseaseFree}.
Consider the
epidemic generated by a single infected individual introduced in a
population at the disease-free equilibrium. This individual makes on
average $R_0$ infectious contacts with other individuals in the
population over the course of its infection. The target of each such
contact has a random susceptibility with expectation $\Sigma$, and thus
the average number of infectious contacts actually leading to a new
infection is $R_0 \Sigma$. As long as its size is small, the outbreak
generated by the original infected individual can be approximated by a
branching process with mean number of offspring $R_0 \Sigma$. This
branching process can only lead to a large outbreak if it is
supercritical, that is, if $R_0 \Sigma > 1$. Therefore,
this criterion expresses that a vaccination policy prevents
endemicity if it prevents a single infected individual in a population at
the disease-free equilibrium from starting a large outbreak.

This interpretation of the threshold is reminiscent of the celebrated
next-generation techniques in epidemic modeling
\citep{diekmann1990definition, diekmann2010construction} for assessing if
a disease can invade a population with heterogeneous susceptibility,
contacts, and infectiousness. However, note that in our model the
susceptibility of an individual is not fixed, but changes as it gets
vaccinated and its immunity is waning, and that the threshold
characterizes the existence of an endemic equilibrium rather than the
possibility of disease invasion. A similar interpretation was also
proposed in \cite{carlsson2020modeling} for their model.

\subsection{Long-term behavior of the solutions}
\label{SS:asymptoticHomogeneous}

The computation in the previous section suggests that an endemic
equilibrium exists if and only if $R_0 \Sigma > 1$. From a public health
perspective, it is important to assess if this endemic equilibrium
corresponds to the long-time behavior of our model when it exists. That
is, we would like to assess the stability of this equilibrium.

In the simple case of the SIRS model (with no vaccination), when an
endemic equilibrium exists (when $R_0 > 1$ in that case) it can be proved
that it is globally asymptotically stable. In our more complicated
setting, studying mathematically the stability of the endemic equilibrium
seems out of reach. We thus investigate this question numerically.

\begin{figure}
    \centering
    \includegraphics[width=\textwidth]{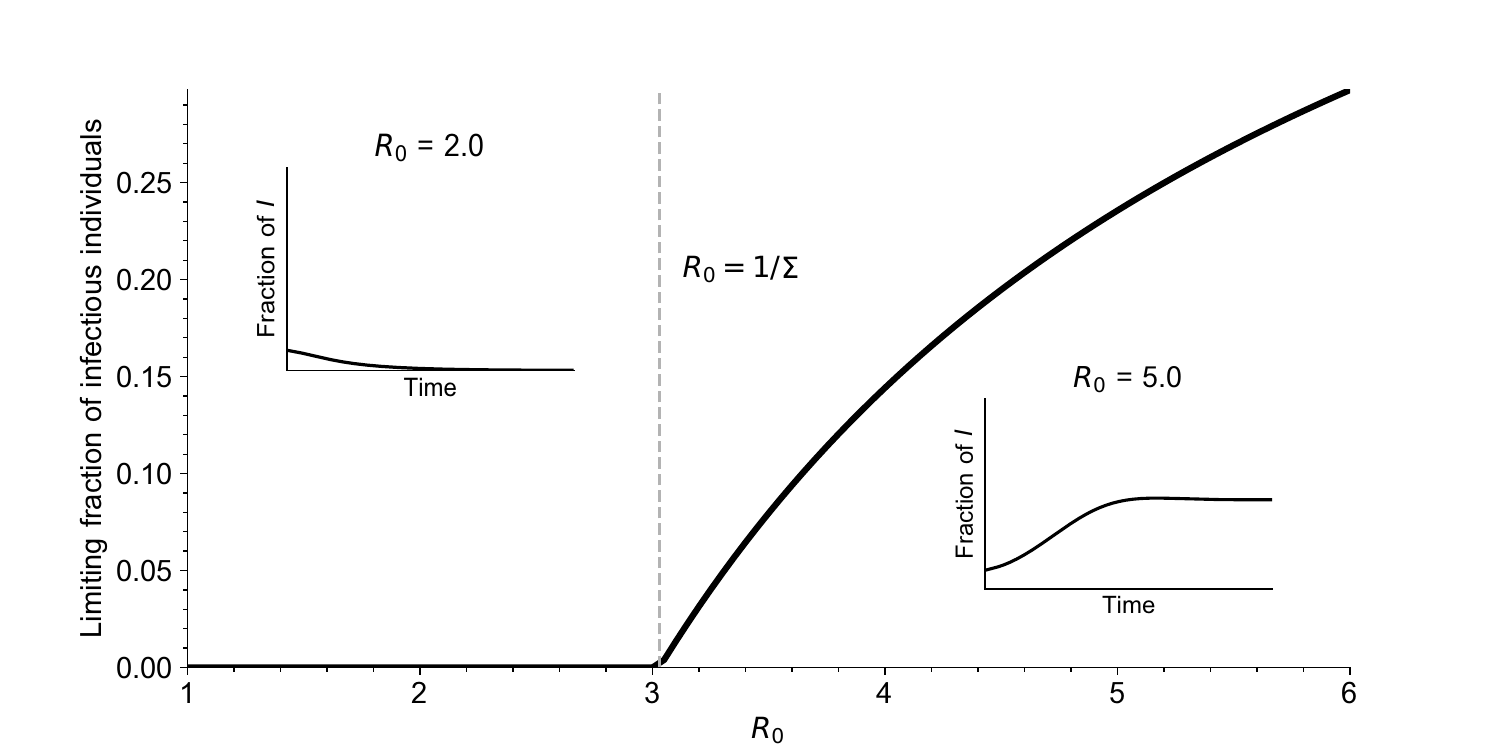}
    \caption{Bifurcation diagram for equation \eqref{eq:main}. For each
    value of $R_0$, the value of $\int_0^\infty I(t,a) \diff a$ is
    reported, for a large time $t = 5000$. The simulations are
    initialized with a fraction $I_0 = 0.1$ of infectious individuals,
    all other parameter values are given in Table~\ref{tab:parametersValue} (Appendix~\ref{sec:numerical:simulations}). 
    The dashed vertical grey line indicates the endemic threshold $1/\Sigma$ computed
    from \eqref{eq:sigma}, above which we expect to see existence of
    a stable endemic equilibrium. In the two insets $\int_0^\infty I(s,a)
    \diff a$ is plotted as a function of time $s\leq t$ for $R_0 = 2$ and $R_0 = 5$.}
    \label{fig:bifurcation1}
\end{figure}

\begin{figure}
    \centering
    \includegraphics[width=\textwidth]{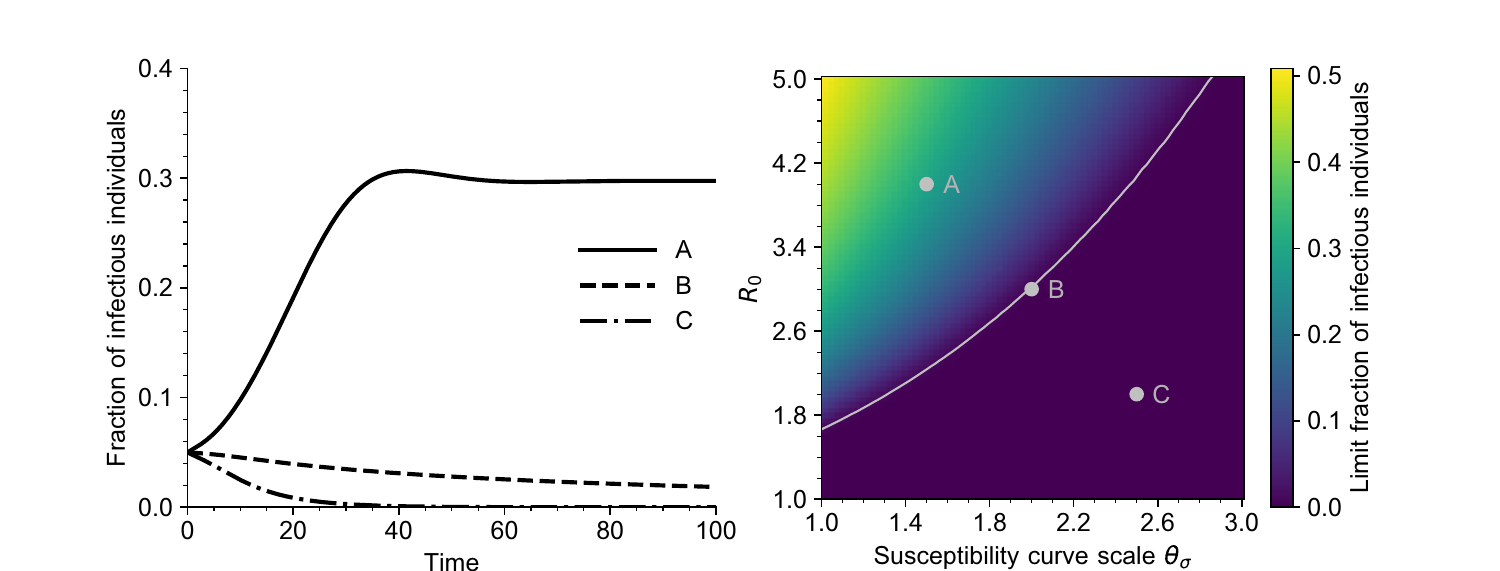}
    \caption{Left: Solutions of the PDE \eqref{eq:main} for
        three values of $R_0$ and $\theta_\sigma$. The parameters
        correspond to the grey dots on the right plot. All other
        parameters are given in Table~\ref{tab:parametersValue} (Appendix~\ref{sec:numerical:simulations}). Right: Bifurcation diagram of equation \eqref{eq:main}, as
        a function of $R_0$ and $\theta_\sigma$ (scale parameter, defined in Section~\ref{sec:model:param}).
        Each point of the heatmap represents the value of $\int_0^\infty
        I(t,a) \diff a$ for a large time $t = 300$. The grey curve is the
        endemic threshold $1/\Sigma$ defined in \eqref{eq:sigma}, as a
        function of $\theta_\sigma$. }
    \label{fig:bifurcation2}
\end{figure}

In Figure~\ref{fig:bifurcation1} and Figure~\ref{fig:bifurcation2} we
draw the bifurcation diagram of our system, as well as some typical
trajectories of the fraction of infected individuals in our model. A
first observation is that, on the simulation displayed in
Figure~\ref{fig:bifurcation1} and in all simulations carried out by the
authors, the model exhibits a simple asymptotic behavior. The epidemic
either dies out, in the sense that the fraction of infected individuals
goes to $0$, or survives in which case the fraction of infected
individuals converges to a positive value. Moreover, the region of
parameters for which the epidemic survives coincides with the region of
parameters for which $R_0 \Sigma > 1$, that is, for which we predict the
existence of an endemic equilibrium.

Overall, this suggests that the asymptotic behavior of our model is very
similar to that of the more usual systems of ordinary differential
equations of the SIRS type: when an endemic equilibrium exists, it is
globally asymptotically stable, otherwise the disease-free equilibrium is
globally asymptotically stable. When $R_0$ crosses the threshold $1/\Sigma$,
we observe an exchange of stability of the two equilibria, similar to a
transcritical bifurcation.

Let us make a final remark. The solution to the PDE \eqref{eq:main} has a probabilistic
interpretation as the age distribution of the solution to the
McKean--Vlasov equation \eqref{eq:McKeanVlasov}. In this probabilistic
setting, existence of an endemic equilibrium translates into existence of
a stationary age distribution, and proving the asymptotic stability of
this equilibrium amounts to proving convergence of the age distribution
towards the stationary distribution. This connexion could provide a way
to study the stability of endemic equilibria analytically. We refer to
models using piecewise deterministic Markov processes with age dependence
such as \cite{Bouguet15, Locherbach-Fournier16} for similar ideas.

\section{Impact of the vaccination policy on endemicity}
\label{S:impactVac}

In the light of the results of the previous section, the long-term
behavior of the epidemic depends mostly on three parameters, namely $R_0$,
$\mathcal{L}_\sigma$ and $\mathcal{L}_V$ (the distributions of $\sigma$
and $T_V$). In this section, we discuss the impact that policy-making can
have on the control of the epidemic through changing these parameters. 

Policy-making does not impact these three parameters in the same way.
The basic reproduction number $R_0$ can be lowered by reducing the
contact rate in the population, but is not dependent on the way vaccines
are administrated. We will consider it as fixed since we are mostly
interested in studying the impact of vaccination rather than changes in
the contact rate. Similarly, we think of $\sigma$ as reflecting the
protection against reinfection provided by the host immunity. The waning
of this protection is therefore dictated by the biological features of
the disease and of the host immunity, which cannot be influenced by
policy-making. (We do neglect the fact that part of the variation of the
susceptibility might come from behavioral changes that could be affected
by policy.) We thus consider the law $\mathcal{L}_\sigma$ of $\sigma$ as
being also fixed. Finally, we think of $T_V$ as resulting from the vaccination
strategy being applied. Typically, the law of $T_V$ depends on the number
of doses administrated, on the instructions given to the general
population on when and how often to get vaccinated, and on how these
instructions are being followed. The law of $T_V$ has a complicated
effect on the outcome of the disease, which depends strongly on the
distribution of $\sigma$ and that we aim to study.

In the rest of this work, we will use $\Sigma$ as an indicator
of the efficiency of the vaccination policy, and try to see what
distribution $\mathcal{L}_V$ of $T_V$ might achieve a lower $\Sigma$.

\subsection{The cost of a vaccination policy}

Intuitively, vaccinating the population more often on average should
result in a higher protection against transmissions, but comes at a
higher cost (of producing the vaccines and deploying them for instance).
We will quantify this cost in order to compare the efficiency of a
vaccination strategy (that is, of a distribution of $T_V$) relative to
its cost, and not only in absolute terms.

A natural measure of the cost of a vaccination policy is the per capita
per unit of time number of vaccine doses that are injected. In our model,
the number of doses injected between time $t$ and $t + \diff t$ is 
\[
    \int_0^\infty S(t,a) \mu_V(a) \diff a \cdot \diff t.
\]
If the population is at the disease-free equilibrium \eqref{eq:diseaseFree}, 
a simple computation shows that the number of doses administrated per
unit of time is
\[
    \int_0^\infty S(a) \mu_V(a) \diff a = \frac{1}{\E[T_V]}.
\]
We argue that, as long as the incidence and prevalence of the disease are low, 
the number of vaccines doses used per unit of time at the endemic
equilibrium can also be approximated by $1/\E[T_V]$.

Let us suppose that the population is at an endemic equilibrium, and
that the incidence is negligible, that is, that $I(0) \ll 1$. Using
\eqref{eq:eqI} and the latter assumption on the incidence, 
\[
    \int_0^\infty I(a) \diff a = I(0) \E[T_I] \ll 1
\]
so that the prevalence of the disease should also be low. From
\eqref{eq:eqS} and using $I(0) \ll 1$ we compute that the number of doses
injected per unit of time at the endemic equilibrium is approximated by
\[
    \int_0^\infty 
    S(0) \mu_V(a) \exp\left(- \int_0^a \mu_V(u) \diff u  \right)
    \E\left[ e^{- I(0) R_0 \int_0^a \sigma(u) \diff u } \right] \diff a 
    \approx S(0).
\]
Using that the prevalence is negligible, we further deduce from
\[
    1 = \int_0^\infty S(a) \diff a + \int_0^\infty I(a) \diff a 
    \approx \int_0^\infty S(0) \exp\left(- \int_0^a \mu_V(u) \diff u  \right) \diff a
    = S(0) \E[T_V]
\]
that the number of doses injected can be approximated by $S(0) \approx 1/\E[T_V]$.

Overall, based on this heuristic computation, we will use $1/\E[T_V]$ as
an indicator of the cost of a vaccination strategy $T_V$.

\subsection{Impact of the vaccination strategy}

We now study the effect that modifying the distribution of $T_V$ has on the
value of $\Sigma$ defined in \eqref{eq:sigma}. By using Fubini's theorem,
let us first re-write the expression for $\Sigma$ as 
\[
    \Sigma = \frac{\E\big[\int_0^{T_V} \sigma(a) \diff a \big]}{\E[T_V]} = \frac{\E[\Phi(T_V)]}{\E[T_V]},
\]
where the deterministic function $\Phi$ is defined as 
\[
    \forall t \ge 0,\quad \Phi(t) = \E\Big[ \int_0^t \sigma(a) \diff a \Big]
    = \int_0^t \E[\sigma(a)] \diff a.
\]

\paragraph{Optimal strategy for a fixed cost.} We assume that only
a fixed number of doses can be administrated per unit of time in the
population, say $1/m$, so that we restrict our attention to random
variables $T_V$ verifying $\E[T_V] = m$. What distribution of $T_V$ then
achieves the smallest value of $\Sigma$? In other words, given that a fixed
daily number of doses are available, how are these doses best distributed
to achieve the highest average immunity level in the population?

It turns out that this question is easily answered analytically. Since
$a \mapsto \sigma(a)$ is a.s.\ nondecreasing, the function $a \mapsto
\Phi(a)$ is convex. Therefore, applying Jensen's inequality we obtain
that 
\[
    \Sigma = \frac{\E[\Phi(T_V)]}{\E[T_V]} \ge \frac{\Phi(m)}{m},
\]
where we recall that we have assumed that $\E[T_V] = m$. We see that the
right-hand side of the previous inequality, $\Phi(m)/m$, is the
susceptibility at the disease-free equilibrium when $T_V = m$ almost
surely. It corresponds to an idealized situation where each individual
gets vaccinated every $m$ unit of time, exactly.
Therefore, given a vaccination strategy $T_V$, a better strategy that
uses the same number of doses is always to let each individual receive
vaccines at evenly spaced moments.

The optimal allocation strategy with $\E[T_V] = m$ is achieved by letting
$T_V = m$ a.s., that is, by letting $T_V$ follow the distribution with the
smallest dispersion. More generally, we argue that a distribution of
$T_V$ which is less dispersed performs better at preventing an endemic
state. Intuitively, if $T_1$ is less dispersed than $T_2$ and both have
the same mean, the distribution of $T_2$ has more mass at larger times.
Loosely speaking, the contribution of large values to the integral of a
convex function is large, and the value of $\E[\Phi(T_2)]$ should be
larger.
Being more
rigorous, one common way to formalize the notion of dispersal is to use
the notion of \emph{convex ordering} \citep[Chapter
3]{shaked2007stochastic}. A random variable $T_1$ is smaller in convex
ordering than another random variable $T_2$ if 
\[
    \forall \phi \text{  convex},\quad \E[\phi(T_1)] \le \E[\phi(T_2)],
\]
which we write as $T_1 \preceq T_2$. Being larger in convex ordering is a
common indicator of larger dispersion. Trivially, if $T_1 \preceq T_2$ and
$\Sigma_1$ and $\Sigma_2$ are the respective stationary susceptibilities,
we have that $\Sigma_1 \le \Sigma_2$. Overall, this indicates that larger
variability in the vaccination times perform worse at preventing the
spread of the disease at the population level.

Note that vaccinating at a fixed time after recovery was also
found to be the optimal vaccination strategy by \cite{khalifi2022extending}
for a related model. Although, the modeling of the vaccination strategies
is different in this article and in our work.

\paragraph{Effect of increasing the vaccination effort.} We now consider
the effect of varying the mean of $T_V$, which can be interpreted as
varying the number of vaccine doses administrated in the population.
Intuitively, one would expect that reducing $\E[T_V]$ (that is,
vaccinating more) also reduces the stationary susceptibility $\Sigma$
(that is, leads to a higher level of immunity).
However, it is not hard to come up with counter-examples where this is
not the case, and no conclusion can be drawn in general.

Nonetheless, it is a reasonable assumption to suppose that increasing the
number of vaccines administrated will not drastically change the shape of
the distribution of $T_V$, but rather modify it in a continuous way. One
way to model this effect is to consider a variable $m \ge 0$ representing the
(inverse of the) vaccination effort. At vaccination effort $m$, the
vaccination period is distributed as $mT_V$, for some fixed random
variable $T_V$. (Increasing the number of vaccines only changes the
\emph{scale} of the distribution.) Clearly, since $\Phi$ is convex, the
function
\[
    m \mapsto \frac{\E[ \Phi(m T_V) ]}{\E[mT_V]}
\]
is increasing. We do recover the expected and intuitive behavior: as $m$
increases, less vaccines are administrated, and $\Sigma$ increases, that
is, the population becomes more susceptible to the disease.

\section{Public health applications}
\label{S:twoGroups}

As an application of our model, we now study in more details two
specific situations. In the first situation, we assume that a fraction of
the population does not get vaccinated. This could reflect among other
examples vaccine hesitancy, impossibility to receive a vaccine, or
unequal access to the vaccination. We want to understand the impact at
the population level of having such a subpopulation that is not
vaccinated, and to derive an expression for the minimal fraction of the
population that needs to be vaccinated to prevent endemicity.

In the second situation, we consider that the population is divided into
two groups, which can represent two distinct physical locations (cities,
countries), or two groups in a heterogeneous population. We assume that a
fixed number of vaccine doses can be administrated per unit of time, due
to resource limitation such as limited vaccine production or deployment.
We investigate the impact of an uneven allocation of these doses between
the two groups.

In both situations the population is no longer homogeneous, in the sense
that it is made of several groups with a different vaccination policy
enforced in each group. We start by making a straightforward extension of
our model to such a heterogeneous population in Section~\ref{SS:heterogeneousModel}. 
We will derive briefly in Section~\ref{SS:endemicHeterogeneous} a
corresponding law of large numbers and criterion for the existence of an
endemic equilibrium, similar to that for the homogeneous model. In
Section~\ref{sec:twoGroups} we provide some general results in the case
of two subpopulations, and we finally our two situations of interest in
Section~\ref{SS:endemicity_2pop}.

\subsection{Modeling vaccination heterogeneity}
\label{SS:heterogeneousModel}

In this section, we consider a population of size $N$ divided into $L$
subgroups. These groups model some heterogeneity in the population such
as age classes, physical locations, or compliance to public health
recommendations. We suppose that these groups mix heterogeneously
according to some contact matrix, modeling for example heterogeneous
social mixing \citep{prem2017projecting, koltai2022reconstructing} or
mobility patterns between different physical areas
\citep{balcan2009multiscale, merler2010role}. We also assume that
different groups are not vaccinated at the same rate, modeling for
instance vaccination policies targeted at specific groups
\citep{Hardt2016}, heterogeneous administration of vaccines
\citep{Perry2021, Mathieu2021}, or differences in beliefs and compliance
to public health recommendations \citep{Hofmann2006, Downs2008,
Lazarus2023}. We make the simplifying assumption that the group to which
an individual belongs does not change during the course of the epidemic,
and that the distribution of the susceptibility and infectiousness curves
do not depend on the group. Let us give a more precise definition of the
dynamics and of the parameters of this extension.


\paragraph{Description.}
Suppose that each of the $N$ individuals in the population now belongs
to one of $L$ groups, labeled by $\ell \in \{ 1, \dots, L \}$. We assume
that individuals remain in the same group at all times. The number of
individuals in group $\ell$ is denoted by $N_\ell$, and we assume that
$N_\ell / N \to p_\ell \in (0, 1)$ as $N \to \infty$. An individual is
now identified by a pair $(\ell, i)$, $\ell$ being its group and $i \le
N_\ell$ its label within group $\ell$. We will denote its infectiousness
and susceptibility at time $t$ by $\lambda^N_{\ell,i}(t)$ and
$\sigma^N_{\ell,i}(t)$ respectively.
Individuals follow the same dynamics as described in
Section~\ref{SS:model} with two modifications: contacts are heterogeneous
between groups and the group of an individual affects its vaccination
rate.

Heterogeneity in the contacts is encoded as a symmetric matrix
$\Gamma =(\gamma_{\ell, \ell'})$, where $\gamma_{\ell,\ell'} =
\gamma_{\ell',\ell} \ge 0$ gives the contact intensity between an
individual of group $\ell$ and one of group $\ell'$. We assume that this
contact matrix does not depend on the size of the population $N$. Note
that $\gamma_{\ell, \ell'}$ corresponds to a contact rate per pair of
individuals, so that the overall contact rate between group $\ell$ and
group $\ell'$ is $\gamma_{\ell, \ell'} N_\ell N_{\ell'}$. The rate at
which individual $(\ell,i)$ gets infected at time $t$ is 
\[
    \sigma^N_{\ell,i}(t) \cdot \sum_{\ell' = 1}^L \gamma_{\ell',\ell} \Lambda^N_{\ell'}(t),
\]
where 
\[
    \Lambda^N_\ell(t) = \frac{1}{N} \sum_{i'=1}^{N_\ell} \lambda^N_{\ell, i'}(t).
\]
In words, each individual $(\ell', i')$ makes an infectious contact with
$(\ell, i)$ at rate $\gamma_{\ell',\ell} \lambda^N_{\ell',i'}(t) / N$, and
such an infectious contact at time $t$ yields an infection with probability
$\sigma^N_{\ell,i}(t)$.

Upon infection, an individual samples an infectious period $T_I$
and an infectiousness $\lambda$ according to the same distribution
$\mathcal{L}_\lambda$ as in the homogeneous model, regardless of its
group. Upon entry in the $S$ state, the susceptibility $\sigma$ of any
individual is also sampled according to the same common distribution
$\mathcal{L}_{\sigma}$ which does not depend on the group. However, an
individual in group $\ell$ samples its waiting time until the next
vaccination according to the distribution of a random variable $T_\ell$
that depends on the group. (Note that we have dropped the $V$ subscript
to ease the notation.) 

This model could be easily made more general by allowing the distribution 
of the infectiousness and susceptibility $\mathcal{L}_\lambda$ and
$\mathcal{L}_\sigma$ to depend on the group. This could represent
a heterogeneous vulnerability to the disease for instance.

\paragraph{Large population size limit.} As in the homogeneous model, we
can derive a law of large numbers limit for the age and state structure
of the epidemic. Let the empirical measure of ages of $I$ and $S$
individuals in group $\ell$ be denoted respectively as 
\[
    \nu_{I,\ell}^N(t) = 
    \frac{1}{N} \sum_{\substack{i = 1,\dots,N_\ell\\C^N_{\ell,i}(t)=I}}
    \delta_{A^N_{\ell,i}(t)},\quad 
    \nu_{S,\ell}^N(t) = 
    \frac{1}{N} \sum_{\substack{i = 1,\dots,N_\ell\\C^N_{\ell,i}(t)=S}}
    \delta_{A^N_{\ell,i}(t)}.
\]
If the initial age structures converge, the above empirical measures
should converge respectively, as $N \to \infty$, to the solution
$(I_\ell(t,a);\, a \ge 0)$ and $(S_\ell(t,a);\, a \ge 0)$ of the
following multidimensional version of equation~\eqref{eq:main},
\begin{align}
\begin{split} \label{eq:mainHeterogeneous}
    \partial_t I_\ell(t,a) + \partial_a I_\ell(t,a) &= - \mu_I(a) I_\ell(t,a) \\
    \partial_t S_\ell(t,a) + \partial_a S_\ell(t,a) &= - \mu_{V,\ell}(a) S_\ell(t,a) 
    - \sum_{\ell'=1}^L \gamma_{\ell',\ell} \Lambda_{\ell'}(t) \cdot
    \E_{t,a;\ell}[\sigma(a)] S_\ell(t,a) \\
    I_\ell(t,0) &= \sum_{\ell' = 1}^L \gamma_{\ell',\ell} \Lambda_{\ell'}(t)
    \cdot \int_0^\infty \E_{t,a;\ell}[\sigma(a)] S_\ell(t,a) \diff a \\
    S_\ell(t,0) &= \int_0^\infty \mu_{V,\ell}(a) S_\ell(t,a) \diff a + \int_0^\infty \mu_I(a) I_\ell(t,a) \diff a.
\end{split}
\end{align}
The initial condition of this system of PDE is a straightforward
extension of that in equation~\eqref{eq:main} and we do not write it down
explicitly. Nevertheless note that the initial condition should fulfill
that
\[
    \forall \ell \le L,\quad \int_0^\infty I_\ell(0,a) + S_\ell(0,a)
    \diff a = p_\ell,
\]
where $p_\ell=\lim_{N\to +\infty}{N_\ell}/{N}.$

In the previous equation, $\mu_{V,\ell}(a)$ denotes the vaccination rate
in group $\ell$, obtained by replacing $T_V$ by $T_{\ell}$ in
\eqref{eq:hazardFunction}, and we define 
\[
    \Lambda_{\ell}(t) = \int_0^\infty I_\ell(t,a) \E[ \lambda(a) \mid T_I > a] \diff a
\]
and 
\begin{multline*}
    \E_{t,a;\ell}[\sigma(a)] = \E\Big[ \sigma(a) 
        \exp\Big(- \int_0^a \sum_{\ell'=1}^L \gamma_{\ell',\ell} \Lambda_{\ell'}(t-a+u) \sigma(u) \diff u\Big) 
    \Big] \\
    \Big/\: \E\Big[ 
        \exp\Big(- \int_0^a \sum_{\ell'=1}^L \gamma_{\ell',\ell} \Lambda_{\ell'}(t-a+u) \sigma(u) \diff u\Big) 
    \Big].
\end{multline*}
All other terms have been defined in Section~\ref{SS:largePopSize}.

\subsection{Endemicity criterion for heterogeneous vaccination}
\label{SS:endemicHeterogeneous}

Again, we study the equilibria of \eqref{eq:mainHeterogeneous} to derive
a criterion for the existence of an endemic equilibrium. We look for
solutions of \eqref{eq:mainHeterogeneous} of the form
\[
    \forall t, a \ge 0,\: \forall \ell \in \{1, \dots, L\},\quad I_\ell(t,a) = I_\ell(a),\quad S_\ell(t,a) =
    S_\ell(a).
\]
We will assume from now on that the matrix $\Gamma$ is irreducible. In
this case, it is not hard to see that there are only two possible types
of equilibria: either the disease is absent in each groups ($I_\ell
\equiv 0$ for all $\ell \in \{1, \dots, L\}$) or it is endemic in each
group ($I_\ell > 0$ for all $\ell \in \{ 1, \dots, L\}$).
Naturally, we will refer to the former situation as a disease-free
equilibrium, and to the latter one as an endemic equilibrium.

\paragraph{Disease-free equilibrium.} As in the homogeneous case, in the
absence of infected individuals the only remaining dynamics are the
vaccination according to renewal processes. It is not hard to check that the only
equilibrium of \eqref{eq:mainHeterogeneous} with $I_\ell(a) \equiv 0$ for all
$\ell \ge 1$ is given by
\[
    \forall a \ge 0,\: \forall \ell \in \{1, \dots, L\},\quad
    S_\ell(a) = \frac{p_\ell}{\E[T_\ell]} \exp\Big(-\int_0^a \mu_{V,\ell}(u) \diff u \Big).
\]

\paragraph{Endemic equilibrium.}
In principle, we could use the same arguments as in the homogeneous case
and find a set of $L$ coupled equations similar to \eqref{eq:endemicity}
that characterize the existence of stationary points of
\eqref{eq:mainHeterogeneous}. However, solving these equations would prove
to be an even more difficult task in this multi-dimensional setting. We
choose not to go in this direction and prefer to start from the
connection between the endemicity criterion in the homogeneous case and
the survival of a well-chosen branching process.

Suppose that a single individual of group $\ell$ is infected in a
population at the disease-free equilibrium. Over its entire infectious
period, this individual makes on average $p_{\ell'} \gamma_{\ell, \ell'}
R_0$ infectious contacts with individuals of type $\ell'$. An individual
of group $\ell'$ targeted by an infectious contact has a random
susceptibility. The expectation of this random variable is the mean
susceptibility at the disease-free equilibrium of group $\ell'$, that is,
\begin{equation} \label{eq:groupSusceptibility}
    \Sigma_{\ell'} 
    \coloneqq \frac{1}{p_{\ell'}} \int_0^\infty S_{\ell'}(a) \E[\sigma(a)] \diff a 
    = \frac{\E\big[ \int_0^{T_{\ell'}} \sigma(a) \diff a\big]}{\E[T_{\ell'}]}.
\end{equation}
Therefore, an infected individual of group $\ell$ produces on average
$R_0 m_{\ell,\ell'}$ secondary infections in group $\ell'$, with
\begin{equation} \label{eq:nextGenMatrix}
    m_{\ell,\ell'} := p_{\ell'} \gamma_{\ell, \ell'} \Sigma_{\ell'}.
\end{equation}
We introduce the matrix
\[
    M \coloneqq (m_{\ell,\ell'})_{1\leq\ell,\ell'\leq L}.
\]

According to the previous discussion, the epidemic generated by a single
infected individual can be thought of as a multi-type branching process
with mean offspring matrix $R_0 M$. The type of an individual in the
branching process corresponds to the group to which it belongs. It is now
a classical result from the theory of branching processes that, under our
mild condition that the contact matrix $\Gamma$ is irreducible, the
latter branching process can survive with positive probability if and
only if the leading eigenvalue of its mean offspring matrix $R_0 M$ is
larger than $1$, that is, if and only if $R_0 \rho > 1$, where $\rho$
is the leading eigenvalue of the matrix $M$ \citep[Chapter~V]{athreya1971branching}. 
This is again reminiscent of the next-generation matrix techniques of
\cite{diekmann1990definition, diekmann2010construction}.

\begin{figure}
    \centering
    \includegraphics[width=\textwidth]{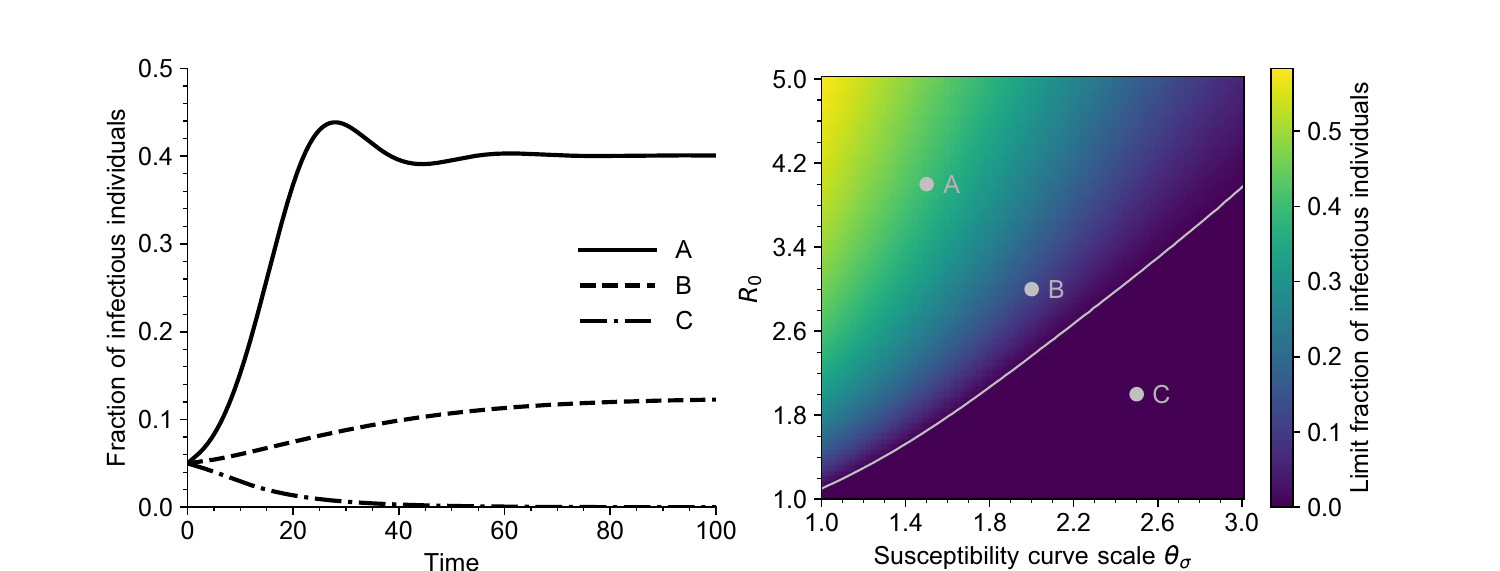}
    \caption{Left: Solutions of the PDE \eqref{eq:main} for
        three values of $R_0$ and $\theta_\sigma$. The parameters
        correspond to the grey dots on the right plot. All other
        parameters are given in Table~\ref{tab:parametersValue} (Appendix~\ref{sec:numerical:simulations}). Right: Bifurcation diagram of equation \eqref{eq:mainHeterogeneous}, 
        as a function of $R_0$ and $\theta_\sigma$ (scale parameter, defined in
        Section~\ref{sec:model:param}). The population is made of three subpopulations with
        contact matrix and vaccination parameters given in Table~\ref{tab:parametersValueHeterogeneous}. 
        Each point of the heatmap represents the value of the total
        fraction of infectious individuals 
        $\int_0^\infty I_1(t,a)+I_2(t,a)+I_3(t,a) \diff a$ for a large time $t = 300$. The
        grey curve is the endemic threshold $1/\rho$ as a function of
        $\theta_\sigma$, where $\rho$ is the leanding eigenvalue
        of $M$ defined in \eqref{eq:nextGenMatrix}.}
        \label{fig:bifucartionHeterogeneous}
\end{figure}

\paragraph{Asymptotic behavior.} As in the case of homogeneous contacts,
we are ultimately interested in assessing the effect of the parameters of
the model on the persistence of the disease in the population for a long
time. The criterion that we have derived above for the existence of an
endemic equilibrium is heuristic, and does not guarantee that the state
of the population converges to that endemic equilibrium when it exists. 

Again, we study the asymptotic behavior of the PDE numerically by considering the
bifurcation diagram of our model in Figure~\ref{fig:bifucartionHeterogeneous}, 
in the case of three subpopulations. The trajectory of the total fraction
of infected individuals among all groups is plotted for a sample of
typical trajectories. As in the homogeneous case, the asymptotic behavior
of the model is simple, and it seems to converge to a limit. In this
limit, depending on the parameters, the epidemic is either extinct or
has reached an endemic equilibrium. Again, we see a good agreement
between our theoretical prediction for the existence of an endemic
equilibrium ($R_0\rho > 1$) and the parameter region where the epidemic
does not go extinct. This also validates that our heuristic, based on the
survival probability of a certain multi-type branching process, seems to
give the right criterion for the existence of an endemic equilibrium.

\subsection{General results on two groups}
\label{sec:twoGroups}

We will discuss our two applications in the simpler context of only two
subpopulations, $L = 2$. Before considering these applications, let us
describe shortly how we parametrize the contact matrix and give some
general results in this case.

The contact matrix $\Gamma$ introduces many new parameters to the model.
We reduce the number of such parameters by assuming that all groups have the
same activity level. That is, we assume that each individual makes on
average contacts at the same rate, regardless of its group. Without loss
of generality, we can assume that this average number of contacts is $1$,
which leads to the constraint that 
\begin{equation} \label{eq:activityLevel}
    p_1 \gamma_{1,1} + p_2 \gamma_{1,2} = 1,\qquad
    p_1 \gamma_{1,2} + p_2 \gamma_{2,2} = 1.
\end{equation}
In the case of a general number of groups $L$ this condition would read
\[
    \forall \ell \in \{1, \dots, L\},\quad \sum_{\ell'=1}^L p_{\ell'} \gamma_{\ell, \ell'} = 1.
\]
For $L=2$, under assumption \eqref{eq:activityLevel} and the additional
constraint that contacts are symmetric, all contact matrices can be
parametrized as
\begin{equation} \label{eq:contact_matrix}
    \Gamma = 
    \begin{pmatrix}
        \frac{1}{p_1}(1-p_2\alpha) & \alpha \\
        \alpha & \frac{1}{p_2}(1-p_1\alpha)
    \end{pmatrix}
\end{equation}
for some $\alpha \in
\left[0,\min\left(\frac{1}{p_1},\frac{1}{p_2}\right)\right]$. The
remaining degree of freedom $\alpha$, which we will refer to as the
contact parameter, tunes the assortativity of the contacts:
\begin{itemize}
    \item for $\alpha \in [0, 1)$ the population is assortative and
        individuals make more contacts within their own group (for $\alpha
        = 0$ the populations would be disconnected);
    \item for $\alpha = 1$  the population is well-mixed and
        contacts are homogeneous;
    \item for $\alpha \in\left(1,\min\left(\frac{1}{p_1},\frac{1}{p_2}\right)\right]$  the population is
        dissortative and individuals make more contacts outside of their
        groups.
\end{itemize}
Under this parametrization, the endemic threshold is given by the
inverse of the leading eigenvalue of the matrix
\[
    M=\begin{pmatrix}
        (1-p_2\alpha)\Sigma_{1}  &   p_2\alpha\Sigma_{2}    \\
        p_1\alpha\Sigma_{1}      &   (1-p_1\alpha)\Sigma_{2}
    \end{pmatrix},
\]
where $\Sigma_1$ and $\Sigma_2$ have been defined in
\eqref{eq:groupSusceptibility}. The leading eigenvalue of this matrix
corresponds to the largest root of the equation
\[
    \rho^2 - ( (1-p_1\alpha)\Sigma_{2} + (1-p_2\alpha)\Sigma_{1})
    \rho + (1-\alpha) \Sigma_1\Sigma_2=0,
\]
given by
\begin{align}\label{eq:rho}
    \rho(\alpha)    
    &=\frac{1}{2}\left(\left(1-p_2\alpha\right)\Sigma_{1}+\left(1-p_1\alpha\right)\Sigma_{2}\right) \nonumber\\ 
    &\hskip 1cm+\frac{1}{2}\sqrt{\left(\left(1-p_2\alpha\right)\Sigma_{1}+\left(1-p_1\alpha\right)\Sigma_{2}\right)^2-4(1-\alpha)\Sigma_{1}\Sigma_{2}}.
\end{align}

Two general observations can be made at this point, which are stated
in Proposition~\ref{prop:rho-alpha} below. First, when $\Sigma_1 = \Sigma_2
\eqqcolon \Sigma$ the leading eigenvalue is $\rho = \Sigma$ and does not
depend on the contact parameter $\alpha$. Therefore, when two groups
are vaccinated in the same way ($\Sigma_1 = \Sigma_2$) and have
the same activity level (equation \eqref{eq:activityLevel} holds), the
population structure does not impact the existence of an endemic
equilibrium. A similar claim holds for any number of groups $L$.
Second, the leading eigenvalue $\rho$ is a non-increasing function of
$\alpha$ when all other parameters are fixed. This indicates that a
population with less assortative contacts performs better at preventing
the disease from reaching an endemic state.

\begin{proposition} \label{prop:rho-alpha}
    If $\Sigma_1 = \Sigma_2$, then $\rho = \Sigma_1$. Moreover, for any
    $\Sigma_1$ and $\Sigma_2$ the function $\alpha  \mapsto
    \rho(\alpha)$ defined in \eqref{eq:rho} is non-increasing and convex
    on $\left[ 0, \min(\tfrac{1}{p_1}, \tfrac{1}{p_2}) \right]$.
\end{proposition}

\begin{proof}
If $\Sigma_1 = \Sigma_2$, $M$ is a multiple of a stochastic matrix and
the leading eigenvalue is easily seen to be $\Sigma_1$. In particular,
the second part of the statement also holds in that case.

Let us now suppose that $\Sigma_1 \ne \Sigma_2$.
Denoting $\Sigma = p_2 \Sigma_1 + p_1 \Sigma_2$, we write $2\rho(\alpha)
= \left(\Sigma_{1}+\Sigma_{2}-\Sigma\alpha\right) + \sqrt{z(\alpha)}$,
with
\[
    z(\alpha)=\left(\Sigma_{1}+\Sigma_{2}-\Sigma\alpha\right)^2+4(\alpha-1)\Sigma_{1}\Sigma_{2}.
\]
Thus $2\rho'(\alpha)=-\Sigma+z'(\alpha)/(2\sqrt{z(\alpha)})$ and $2\rho''(\alpha)=\left(2z''(\alpha)z(\alpha)-z'(\alpha)^2\right)/\left(4z(\alpha)^{3/2}\right)$. We have
\begin{align*}
  z'(\alpha)&=-2\Sigma\left(\Sigma_{1}+\Sigma_{2}-\Sigma\alpha\right)+4\Sigma_{1}\Sigma_{2}\\
  z''(\alpha)&=2{\Sigma}^2.
\end{align*}
Consequently,
\begin{align*}
   2z''(\alpha)z(\alpha)-z'(\alpha)^2
   &=16\Sigma_{1}\Sigma_{2}\left(-\Sigma^2-\Sigma_{1}\Sigma_{2}
   +\Sigma(\Sigma_{1}+\Sigma_{2})\right)\\
   &=16p_1p_2(\Sigma_{1}-\Sigma_{2})^2> 0.
\end{align*}
We deduce that $\alpha\mapsto \rho'(\alpha)$ is an increasing  function. 

By noting that 
\[
    z(\alpha) = \left((1-p_2\alpha)\Sigma_{1} - (1-p_1\alpha)\Sigma_{2}\right)^2
    + 4\alpha^2 p_1 p_2 \Sigma_1 \Sigma_2 > 0,
\]
we see that $\alpha \mapsto \rho(\alpha)$ is well-defined on $[0, \infty)$ 
and that the previous computation still holds. Moreover, $\lim_{\alpha\to
\infty}\rho'(\alpha)=0$, and $\rho'$ only takes negative values. We
conclude that $\rho$ is a decreasing convex function.
\end{proof}

\subsection{Two public health applications}
\label{SS:endemicity_2pop}

We now study our two situations of interest.

\paragraph{Effect of vaccine hesitancy.}
We model partial vaccination of the population by assuming that
individuals from group $1$ get vaccinated whereas individuals from group
$2$ do not. Let $\Sigma \coloneqq \Sigma_1$ be the mean susceptibility at
the disease-free equilibrium within group $1$. We think of
$\Sigma$ as being fixed, corresponding to a given vaccination policy, and
$p_1$ as varying depending on the fraction of the population complying
with this policy. Let us assume that
almost surely $\sigma(a) \to 1$ as $a \to \infty$, so that individuals
immunity vanishes completely after a long-enough time. Since group $2$ does not
get vaccinated, the mean susceptibility in this group is set to be $\Sigma_2 = 1$.

We further assume for simplicity that the population is well-mixed
($\alpha = 1$), so that the mean offspring matrix is 
\[
    M =
    \begin{pmatrix}
        p_1 \Sigma_1 & p_2  \\
        p_1 \Sigma_1 & p_2 
    \end{pmatrix}
\]
and we can readily check that its leading eigenvalue is 
\[
    \rho = p_1 \Sigma + p_2 = 1 - (1-\Sigma)p_1.
\]
Define a critical fraction $p_c$ as 
\begin{equation} \label{eq:criticalFraction}
    p_c \coloneqq \frac{1-1/R_0}{1-\Sigma}.
\end{equation}
Then $p_c$ gives the critical fraction of the population that needs to be
vaccinated recurrently to prevent an endemic equilibrium, that is
\[
    R_0\rho \le 1 \iff p_1 \ge p_c.
\]
There is an interesting correspondence between this formula and the
well-known formula that gives the critical vaccine coverage to prevent an
epidemic \citep{anderson1982directly, anderson2020challenges}. If a
vaccine has an efficacy $E \in [0,1]$ (that is, if it provides a sterilizing
immunity with probability $E$) then the critical fraction of the
population that needs to be vaccinated to prevent an epidemic is 
\[
    p'_c = \frac{1-1/R_0}{E}.
\]
In our model, the efficiency of the vaccine policy is quantified by
$1-\Sigma$ which corresponds to the fraction of infections that are
blocked at the stationary disease-free equilibrium if all individuals get
vaccinated.

\paragraph{Optimal vaccine allocation between two groups.}
Consider a second situation where a fixed number of vaccine doses per
unit of time is available, and these doses need to be allocated between
two groups of individuals, which do not necessarily make homogeneous
contacts. (The contact heterogeneity accounts for the fact that the
groups may be two physically distinct locations: cities, countries,
regions.) We model this situation in the following way.

Let $T_1$ and $T_2$ be the vaccination times in each group with expectations $m_1 =
\E[T_1]$ and $m_2 = \E[T_2]$, and let $1/m$ be the per unit of time
number of doses that can be allocated in the total population. Since the number
of doses injected in group $\ell$ is $p_\ell / \E[T_\ell]$, the fact that the
total number of doses injected in the population is $1/m$ adds the
constraint that 
\begin{equation} \label{eq:fixedDoses}
    \frac{p_1}{\E[T_1]} + \frac{p_2}{\E[T_2]} 
    = \frac{p_1}{m_1} + \frac{p_2}{m_2}
    = \frac{1}{m}.
\end{equation}
The set of all pairs $(m_1, m_2)$ verifying \eqref{eq:fixedDoses} for a
given $m$ can now be parametrized by a single parameter $\beta$ as
\[
    \forall \beta\in\left[-\frac{1}{p_2},\frac{1}{p_1}\right], \quad \frac{1}{m_1(\beta)}=\frac{1}{m}+\frac{ p_2}{m}\beta,\quad \frac{1}{m_2(\beta)}=\frac{1}{m}-\frac{ p_1}{m}\beta.
\]
Under this parametrization, we can interpret $\beta$ as assessing the
\emph{fairness} of the allocation in the sense that
\begin{itemize}
    \item when $\beta=0$, all doses are allocated evenly across the two groups;
    \item when $\beta>0$ population 1 is favored and if $\beta=\frac{1}{p_1}$ all the doses are allocated to population $1$;
    \item when $\beta<0$ population 2 is favored and if $\beta=-\frac{1}{p_2}$, all the doses are allocated to population $2$.
\end{itemize}

Fix some random variable $T$ with $\E[T] = m$ and define 
\begin{equation} \label{eq:vaccinationAllocation}
    T_1(\beta) = \frac{1}{1+\beta p_2} T,
    \qquad
    T_2(\beta) = \frac{1}{1-\beta p_1} T.
\end{equation}
Then $(T_1(\beta), T_2(\beta))$ is a natural family of random variables
verifying that $\E[T_\ell(\beta)] = m_\ell(\beta)$, and represents a
possible allocation of the doses between the two population with fairness
parameter $\beta$.
We show below in Proposition~\ref{prop:fairAllocation} that, when the
population is assortative ($\alpha \le 1$), the minimal eigenvalue $\rho$
is achieved at $\beta = 0$. In other words, the best possible allocation
to prevent an endemic state is the fair allocation ($\beta = 0$) where
individuals in both subpopulations receive the same amount of vaccines.
This result is illustrated in Figure~\ref{fig:Rc-2pop}.

\begin{figure}
    \centering
    \includegraphics[width=\textwidth]{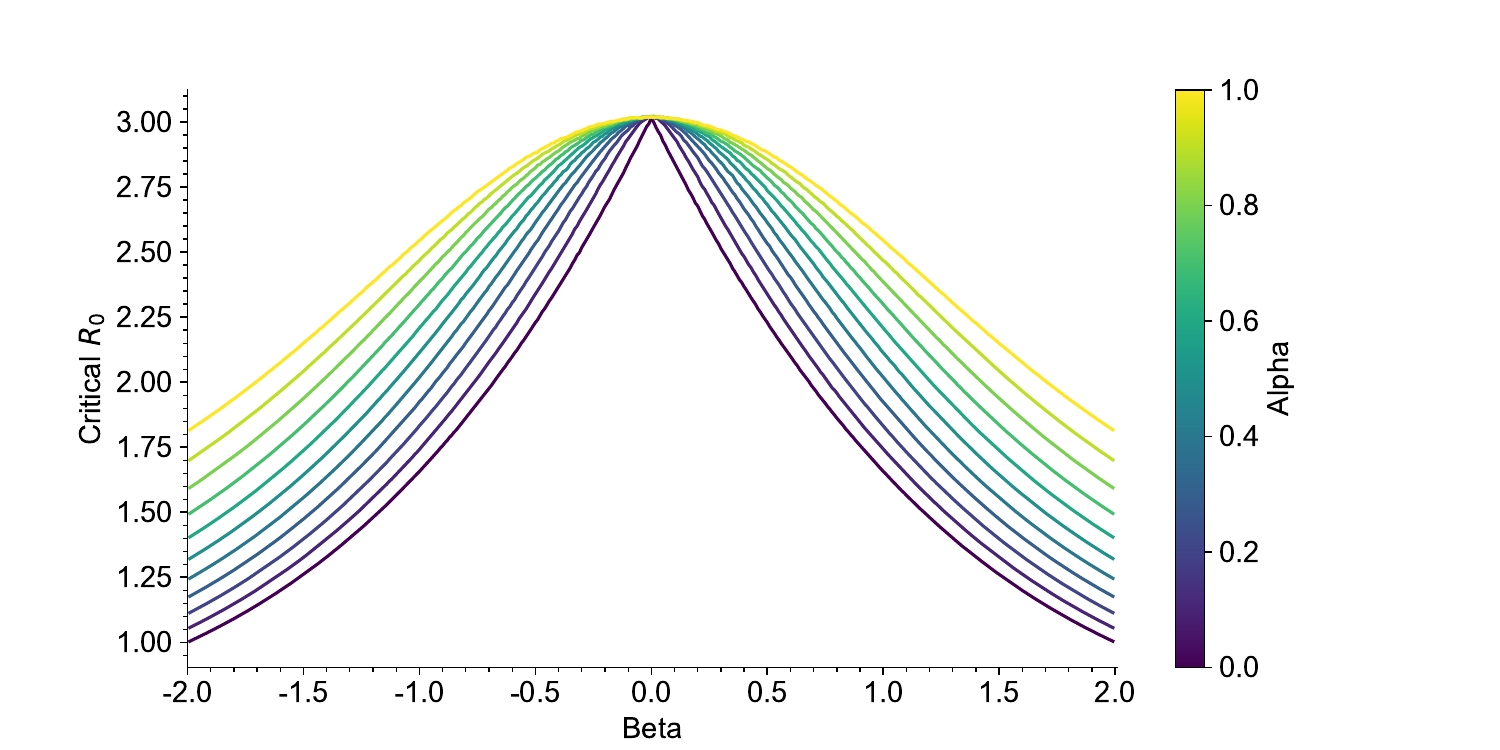}
    \caption{Behavior of $1/\rho$ as a function of the fairness parameter
        $\beta$, for different values of $\alpha$. The random variable 
        $T$ in \eqref{eq:vaccinationAllocation} has a Gamma distribution
        with shape and scale parameters given by $\kappa_V$ and $\theta_V$
        respectively, as in Table~\ref{tab:parametersValue} (Appendix~\ref{sec:numerical:simulations}). The
        population is assumed to be made of two groups of the same size,
        $p_1 = p_2 = \tfrac{1}{2}$. All other parameters are given in
        Table~\ref{tab:parametersValue} (Appendix~\ref{sec:numerical:simulations}).}
    \label{fig:Rc-2pop}
\end{figure}

\begin{proposition} \label{prop:fairAllocation}
    Let $(T_1(\beta), T_2(\beta))$ be as in \eqref{eq:vaccinationAllocation} 
    and let $\rho(\beta)$ be the largest eigenvalue of \eqref{eq:nextGenMatrix}.
    Then for any fixed number of doses $m > 0$ and any $\alpha \in
    [0,1]$, $\beta \mapsto \rho(\beta)$ is minimal at $\beta = 0$.
\end{proposition}

\begin{proof}
Let us write $\rho(\alpha, \beta)$ to emphasize the dependence in the
two parameters. We need to show that for any $\alpha \in [0, 1]$,
$\rho(\alpha, \beta) \ge \rho(\alpha, 0)$. From Proposition~\ref{prop:rho-alpha}, 
for any $\alpha \in [0,1]$ and $\beta \in [\tfrac{-1}{p_2},\tfrac{1}{p_1}]$
\[
    \rho(\alpha, \beta) \ge \rho(1, \beta) \coloneqq p_1 \Sigma_1(\beta) + p_2 \Sigma_2(\beta),
\]
where, for $i \in \{1, 2\}$, $\Sigma_i(\beta) = \E[\phi(T_i(\beta))] / \E[T_i(\beta)]$,
and $\phi \colon a\mapsto\int_0^a\sigma(u)\mathrm{d}u$. According
to \eqref{eq:fixedDoses} we have
\[
    \frac{p_1}{\E[T_1(\beta)]} + \frac{p_2}{\E[T_2(\beta)]} = \frac{1}{\E[T]},
\]
and we can use that $\phi$ is convex to obtain that 
\begin{align*}
    \rho(1, \beta) 
    &= 
    p_1 \frac{\E[\phi(T_1(\beta))]}{\E[T_1(\beta)]} + 
    p_2 \frac{\E[\phi(T_2(\beta))]}{\E[T_2(\beta)]} \\
    &\ge
    \frac{1}{\E[T]} \E\Big[ \phi\Big(\frac{p_1\E[T]}{\E[T_1(\beta)]} T_1(\beta)
    + \frac{p_2\E[T]}{\E[T_2(\beta)]} T_2(\beta)\Big) \Big] \\
    &= \frac{\E[\phi(T)]}{\E[T]} = \rho(1, 0) = \rho(\alpha, 0).
\end{align*}
In the last line, we have again used that $M$ is a multiple of a
stochastic matrix when $\beta = 0$, so that $\rho(\alpha, 0)$ is
constant. Overall, we have shown that $\rho(\alpha, \beta) \ge
\rho(\alpha, 0)$.
\end{proof}

\section{Well-posedness of the PDE system}
\label{sec:PDE}

In this section we provide proofs related to the solution of the PDE
system~\eqref{eq:main}. We first prove in Section~\ref{sec:proof-wellposedness} 
the existence and uniqueness of the solution. In Section~\ref{sec:proofPDE}, 
we prove Proposition~\ref{prop:PDErepresentation}, which identifies the
solution  as the limit of the stochastic model when the size of the
population goes to infinity.

\subsection{Proof of Proposition~\ref{prop:uniqueness-sol-PDE}}
\label{sec:proof-wellposedness}



We recall that the  Skorokhod space $\mathbb{D}(\R^+,\R^+)$ is the space
of right continuous with left limits functions on $\R^+$ with values in
$\R^+$ (see \cite{billingsley} for more details).

Using the definition of a weak solution, we can reformulate equation
\eqref{eq:main} as a set of Volterra equations. The force of infection
$\Lambda$ is defined by \eqref{eq:Lambda}, 
\[
    \forall t \ge 0,\quad \Lambda(t) = \int_0^{\infty} \E[\lambda(a) \mid T_I>a]I(t,a)\diff a,
\]
and the  mean susceptibility of the population by
\begin{equation}
    \forall t \ge 0,\quad \Sigma(t)=\int_0^{\infty}\E_{t,a}[\sigma(a)]S(t,a)\diff a,
\end{equation}
where $\E_{t,a}[\sigma(a)]$ is given by \eqref{eq:biasExpectation}.
By \eqref{eq:main}, for $t\geq0$, $I(t,0)=\Lambda(t)\Sigma(t)$ and 
\[
S(t,0)= \int_0^\infty \mu_I(a) I(t,a) \diff a + \int_0^\infty \mu_V(a) S(t,a) \diff a.
\]
An individual of age $a$ at time $t$ was of age $0$ at time $t-a$ and no new event (recovery, vaccination, infection) occurred between $t-a$ and $t$. We thus deduce
\begin{align*}
\Lambda(t)&=\int_0^t\E[\lambda(a)] I(t-a,0)\diff a+I_0\int_t^{\infty}\E[\lambda(a)]h_I(a-t)\diff a
\end{align*}
and
\begin{align*}
    \Sigma(t)
    &= \int_0^t 
       \E\left[ \sigma(a) \mathrm{e}^{-\int_0^a \Lambda(t-a+u) \sigma(u) \diff u} \right]
       e^{-\int_0^a\mu_V(u)\diff u}S(t-a,0) \diff a \\
    &+ \int_t^\infty 
       \E\left[\sigma(a)\mathrm{e}^{-\int_{a-t}^a \Lambda(t-a+u) \sigma(u) \diff u} \right]
       e^{-\int_{a-t}^a \mu_V(u) \diff u} (1-I_0) h_S(a-t) \diff a
\end{align*}
where we used $\E[\lambda(a)\indic_{\{T_I>a\}}]=\E[\lambda(a)]$ in the
expression of $\Lambda$, $\Lambda\equiv 0$ on the negative values and
the vaccination rate for the initial susceptible individuals is
$\mu_V(u)\indic_{\{u>a-t\}}$ for individuals of age $a\in(t,+\infty)$ in
the expression of $\Sigma$ (see Section~\ref{SS:formalModel}). 

Using the same arguments for $S(t,0)$, we remark that the pair
$(I(t,0),S(t,0); t\geq 0)$ is solution of the system of integral
equations defined by, for $t\geq 0$,
\begin{align}\label{eq:IES}
\begin{split}
    x(t) &= L(t) \int_0^{\infty} \E\left[ \sigma(a) e^{-\int_0^a L(t-a+u)\sigma(u)\diff u} \right]
         e^{-\int_{(a-t)\vee 0}^a \mu_V(u) \diff u} y(t-a) \diff a, \\
    y(t) &= \int_0^\infty \mu_I(a) e^{-\int_{(a-t)\vee 0}^a \mu_I(u) \diff u} x(t-a) \diff a \\
         &\qquad + \int_0^\infty \mu_V(a) e^{-\int_{(a-t)\vee 0}^a \mu_V(u) \diff u} 
            \E\left[ e^{-\int_0^a L(t-a+u) \sigma(u) \diff u} \right] y(t-a) \diff a,
\end{split}
\end{align}
with $x(t)=I_0h_I(-t)$, $y(t)=(1-I_0)h_S(-t)$ for $t<0$ and $L(t)=\int_0^\infty \E[\lambda(a)]x(t-a)\diff a$ for $t\geq0$, 
$L(t)=0$ for $t<0$.
We observe that $(x,y)$ is a solution of \eqref{eq:IES} if and only if 
\begin{align*}
    I(t,a)&=\mathrm{e}^{-\int_{(a-t)\vee 0}^a\mu_I(u)\diff u}x(t-a)\\
    S(t,a)&=\mathrm{e}^{-\int_{(a-t)\vee 0}^a\mu_V(u)\diff u} \E\left[\mathrm{e}^{-\int_0^a L(t-a+u)\sigma(u)\diff u}\right] y(t-a)
\end{align*}
is a weak solution of \eqref{eq:main}.

\medskip \noindent
\textbf{A priori estimates.}
Let $(x,y)$ be nonnegative functions, solution  of the system
\eqref{eq:IES}. We introduce
\begin{align*}
    z(t)=&\int_0^\infty\mathrm{e}^{-\int_{(a-t)\vee 0}^a\mu_I(u)\diff u}x(t-a) \diff a\\
    &\hskip .5cm+\int_0^{\infty}\E\left[\mathrm{e}^{-\int_0^aL(t-a+u)\sigma(u)\diff u}\right]\mathrm{e}^{-\int_{(a-t)\vee 0}^a\mu_V(u)\diff u}y(t-a)\diff a,
\end{align*}
which can also be written, using the changes of variables $b=t-a$ on
$[0,t]$ and $b=a-t$ on $[0,\infty)$,
\begin{align*}
    z(t)&=\int_0^t\mathrm{e}^{-\int_0^{t-b}\mu_I(u)\diff u}x(b) \diff b+I_0\int_0^\infty\mathrm{e}^{-\int_b^{b+t}\mu_I(u)\diff u}h_I(b) \diff b\\
    &+\int_0^{t}\E\left[\mathrm{e}^{-\int_0^{t-b}L(b+u)\sigma(u)\diff u}\right]\mathrm{e}^{-\int_0^{t-b}\mu_V(u)\diff u}y(b)\diff b
    \\
    &+(1-I_0)\int_0^{\infty}\E\left[\mathrm{e}^{-\int_b^{b+t}L(u-b)\sigma(u)\diff u}\right]\mathrm{e}^{-\int_b^{b+t}\mu_V(u)\diff u}h_S(b)\diff b.
\end{align*}
Computing the first derivative of $z$, we observe that $z'(t)=0$.
Computing $z(0)$, we then have, $\forall t\geq 0$,
\begin{multline}\label{eq:z=1}
    1=\int_0^\infty\mathrm{e}^{-\int_{(a-t)\vee 0}^a\mu_I(u)\diff u}x(t-a) \diff a \\
    +\int_0^{\infty}\E\left[\mathrm{e}^{-\int_0^aL(t-a+u)\sigma(u)\diff u}\right]\mathrm{e}^{-\int_{(a-t)\vee 0}^a\mu_V(u)\diff u}y(t-a)\diff a.
\end{multline}

As $\sigma \in [0, 1]$, we easily deduce from \eqref{eq:IES} and the
above equation that $x(t) \leq L(t)$. Moreover, by assumption
$\E[\lambda(a)] \le \lambda_{\max}$. Consequently, by
definition of $L$, we have
\[
    x(t)\leq \lambda_{\max} \left(\int_0^t x(a)\diff a+I_0\right).
\]
Using Gronwall's Lemma, we obtain for $t \geq 0$, 
$x(t) \leq I_0 \lambda_{\max} e^{\lambda_{\max} t}$
and thus $L(t)\leq I_0 \lambda_{\max} e^{\lambda_{\max} t}$. 

Let $T>0$. Since the density distribution function of $T_V$ is locally
bounded, there exists $C_T > 0$ such that, $\forall t\geq 0$
\begin{align*}
  y(t)&\leq \max_{t\in[0,T]}x(t)+ C_T \int_0^t y(b)\diff b\\
&+\int_0^\infty \mu_I(t+b)\mathrm{e}^{-\int_b^{b+t}\mu_I(u)\diff u}h_I(b)\diff b
+\int_0^\infty \mu_V(t+b)\mathrm{e}^{-\int_b^{b+t}\mu_V(u)\diff u}h_S(b)\diff b.
\end{align*}
Using Gronwall's inequality, {and assumptions of the proposition,} we conclude that $y$ is locally bounded on
$\R^+$.

\medskip\noindent
\textbf{Existence and uniqueness of solutions.}
We now prove the uniqueness of the solution.
Let $(x_1,y_1)$ and $(x_2,y_2)$ be two solutions of \eqref{eq:IES}. 

From \eqref{eq:IES} we have 
\begin{align*}
    &\abs{x_1(t) - x_2(t)}  
    \le \abs{L_1(t) - L_2(t)} 
    \int_0^{\infty} \E\left[ \sigma(a) e^{-\int_0^a L_1(t-a+u)\sigma(u)\diff u} \right]
         e^{-\int_{(a-t)\vee 0}^a \mu_V(u) \diff u} y_1(t-a) \diff a \\
    &\qquad+ L_2(t)
    \int_0^{\infty} \E\left[ \sigma(a) \abs*{e^{-\int_0^a L_1(t-a+u)\sigma(u)\diff u} 
         - e^{-\int_0^a L_2(t-a+u)\sigma(u)\diff u} }\right]
         e^{-\int_{(a-t)\vee 0}^a \mu_V(u) \diff u} y_1(t-a) \diff a \\
    &\qquad+ L_2(t)
    \int_0^{\infty} \E\left[ \sigma(a) e^{-\int_0^a L_1(t-a+u)\sigma(u)\diff u} \right]
        e^{-\int_{(a-t)\vee 0}^a \mu_V(u) \diff u} \abs{y_1(t-a)-y_2(t-a)} \diff a \\
    &\le \lambda_{\max} \int_0^t \abs{x_1(a) - x_2(a)}  \diff a 
    + L_2(t) \int_0^{\infty} \int_0^a \abs*{L_1(t-a+u) - L_2(t-a+u) }\diff u\: y_1(t-a) \diff a \\
    &\qquad+ L_2(t) \int_0^t \abs{y_1(a)-y_2(a)} \diff a.
\end{align*}
For the first term we have used Equation \eqref{eq:z=1} to bound the
integral, and that $L_1(t) - L_2(t) = \int_0^t \E[\lambda(a)](x_1(a)-x_2(a)) \diff a$. 
For the second term we have used that $\abs{e^{-u}-e^{-v}} \le
\abs{u-v}$. For the third time we have used that $y_1(t) = y_2(t)$ for $t
< 0$. We further bound the second term by noting that 
\begin{align*}
    \int_0^\infty &\int_0^a \abs{L_1(t-a+u) - L_2(t-a+u)} \diff u \: y_2(t-a) \diff a\\
    &= \int_0^t \int_{t-a}^t \abs*{L_1(v) - L_2(v)} \diff v \: y_1(t-a) \diff a 
    + \int_0^t \abs*{L_1(u) - L_2(u)} \diff u \cdot \int_t^\infty y_1(t-a) \diff a  \\
    &\le \int_0^t \abs*{L_1(u) - L_2(u)} \diff u \cdot 
    \Big( t \sup_{s \in [0,t]} \abs{y_1(s)} + (1-I_0) \big).
\end{align*}
Using the expression of $L(t)$ and Fubini's theorem we have
\begin{equation} \label{eq:boundIntL}
    \int_0^t \abs{L_1(u)-L_2(u)} \diff u 
    \leq t \lambda_{\max} \int_0^t \abs{x_1(b)-x_2(b)} \diff b.
\end{equation}
Therefore, combining all the previous estimates we see that there exists 
$C_T$ such that for $t \le T$, 
\begin{equation} \label{eq:estimateX}
    \abs{x_1(t) - x_2(t)} 
    \le C_T \Big(\int_0^t \abs{x_1(s) - x_2(s)} \diff s + \int_0^t
    \abs{y_1(s) - y_2(s)}\Big) \diff s.
\end{equation}

In a similar way, \eqref{eq:IES} yields that 
\begin{align*}
    &\abs{y_1(t) - y_2(t)}
    \le \int_0^\infty \mu_I(a) e^{-\int_{(a-t)\vee 0}^a \mu_I(u) \diff u}
    \abs{x_1(t-a)-x_2(t-a)} \diff a \\
    &\qquad + \int_0^\infty \mu_V(a) e^{-\int_{(a-t)\vee 0}^a \mu_V(u) \diff u} 
            \E\left[ \abs*{e^{-\int_0^a L_1(t-a+u) \sigma(u) \diff u} -
                e^{-\int_0^a L_2(t-a+u) \sigma(u) \diff u}} \right] y_1(t-a) \diff a \\
    &\qquad + \int_0^\infty \mu_V(a) e^{-\int_{(a-t)\vee 0}^a \mu_V(u) \diff u} 
            \E\left[ e^{-\int_0^a L_1(t-a+u) \sigma(u) \diff u} \right] \abs{y_1(t-a) - y_2(t-a)} \diff a \\
    &\le \int_0^t \mu_I(t-b)e^{-\int_0^{t-b} \mu_I(u) \diff u} \abs{x_1(b)-x_2(b)} \diff b \\
    &\qquad+ \int_0^t \mu_V(a) e^{-\int_0^a \mu_V(u) \diff u} 
    \int_{t-a}^t \abs*{L_1(b) - L_2(b)} \diff b \: y_1(t-a) \diff a \\
    &\qquad + \int_0^t \abs*{L_1(b) - L_2(b)} \diff b \cdot
    (1-I_0) \int_0^\infty \mu_V(t+b) e^{-\int_b^{t+b} \mu_V(u) \diff u} 
    h_S(b) \diff a  \\
    &\qquad+ \int_0^t \mu_V(t-b) e^{-\int_{t-b}^t \mu_V(u) \diff u} \abs{y_1(b) - y_2(b)} \diff b.
\end{align*}
For the first and third terms we have used that $x_1(t) = x_2(t)$ and
$y_1(t) = y_2(t)$ for $t < 0$. For the second term we have split the
integrals for $a > t$ and $a \le t$.

Our assumptions entail that $\mu_V(a) e^{-\int_0^a \mu_V(u) \diff u}$ and
$\mu_I(a) e^{-\int_0^a \mu_I(u) \diff u}$ are bounded. Therefore, using the previous
inequality, this bound, our assumption \eqref{eq:boundedIC} on the
contribution of the initial individuals together with \eqref{eq:boundIntL}
yield that there exists $C'_T$ such that, for $t \le T$,
\begin{equation} \label{eq:estimateY}
    \abs{y_1(t) - y_2(t)} \le C'_T \int_0^t \abs{x_1(s) - x_2(s)} + \abs{y_1(s) - y_2(s)} \diff s
\end{equation}
The estimates \eqref{eq:estimateX} and \eqref{eq:estimateY} on $x$ and $y$ 
combined with Gronwall's inequality show that $x_1(t) = x_2(t)$ and
$y_1(t) = y_2(t)$ for all $t \ge 0$, proving uniqueness of the solution
to \eqref{eq:IES}. The existence of a solution is proved by a classical Picard method. Let $T>0$ be fixed. For $n\geq 0$, we define by induction the sequences $(L_n)_{n\geq 0}$, $(x_n)_{n\geq 0}$ and $(y_n)_{n\geq 0}$: for $t\in [0,T]$
\begin{align*}
L_0(t)&=I_0\int_0^\infty\E[\lambda(b+t)]h_I(b)\diff b\\
x_0(t)&=I_0L_0(t)\int_0^\infty e^{-\int_{b}^{b+t}\mu_V(u)\diff u}\E\left[\sigma(b+t)e^{-\int_{0}^{t}L_0(u)\sigma(u)\diff u}\right]h_I(b)\diff b\\
y_0(t)&=I_0\int_0^{\infty}\mu_I(b+t)e^{-\int_{b}^{b+t}\mu_I(u)\diff u}h_I(b)\diff b\\
&\hskip 1cm
+(1-I_0)\E\left[e^{-\int_{0}^{t}L_0(u)\diff u}\right]\int_0^\infty\mu_V(b)e^{-\int_{b}^{b+t}\mu_V(u)\diff u}h_S(b)\diff b\\
L_{n+1}(t)&=I_0\int_0^\infty\E[\lambda(b+t)]h_S(b)\diff b+\int_0^t\E[\lambda(a)]x_n(t-a)\diff a\\
x_{n+1}(t)&=I_0L_{n+1}(t)\int_0^\infty \E\left[\sigma(b+t)e^{-\int_{0}^{t}L_{n+1}(u)\sigma(u)\diff u}\right]e^{-\int_{b}^{b+t}\mu_V(u)\diff u}h_I(b)\diff b\\&\hskip 1cm +L_{n+1}(t)\int_0^{t}\E\left[\sigma(a)e^{-\int_0^{a}L_{n+1}(t-a+u)\sigma(u)\diff u}\right]e^{\int_{0}^{a}\mu_V(u)\diff u}y_n(t-a)\diff a\\
y_{n+1}(t)&=I_0\int_0^{\infty}\mu_I(b+t)e^{-\int_{b}^{b+t}\mu_I(u)\diff u}h_I(b)\diff b
\\
&+(1-I_0)\E\left[e^{-\int_{0}^{t}L_{n+1}(u)\diff u}\right]
\int_0^\infty\mu_V(b+t)e^{-\int_{b}^{b+t}\mu_V(u)\diff u}h_S(b)\diff b\\
&\hskip -1cm +\int_0^{t}\mu_I(a)e^{-\int_{0}^{a}\mu_I(u)\diff u}x_n(t-a)\diff a
+\int_0^t\mu_V(a)e^{-\int_{0}^a\mu_V(u)\diff u}\E\left[e^{-\int_{0}^aL_{n+1}(t-a+u)\diff u}\right]y_n(t-a)\diff a.
\end{align*}
By iteration and  using Equation \eqref{eq:estimateY}, we prove that 
\begin{multline*}
|x_{n+1}(t)-x_n(t)|+|y_{n+1}(t)-y_n(t)|  \\
\leq C_T^n\int_0^t\int_0^{t_{n-1}}\cdots\int_0^{t_1}|x_{1}(a)-x_0(a)|+|y_{1}(a)-y_0(a)|\diff a\diff t_1\ldots \diff t_{n-1}
\end{multline*}
and then, denoting by $\|.\|_{[0,T]}$ the uniform distance  on the interval $[0,T]$,
\[
\|x_{n+1}-x_n\|_{[0,T]}+\|y_{n+1}-y_n\|_{[0,T]}\leq \frac{C_T^nT^n}{n!}\left(\|x_{1}-x_0\|_{[0,T]}+\|y_{1}-y_0\|_{[0,T]}\right).
\]
The upper-bound is the general term of a converging series,  and we deduce that the sequences $(x_n)_{n\geq 0}$ and $(y_n)_{n\geq 0}$ converge  on the interval $[0,T]$ to a solution of \eqref{eq:IES}. We proved existence and uniqueness of a solution to \eqref{eq:IES} on the interval $[0,T]$ for any $T>0$, we then deduce the existence and uniqueness on $\R^+$.


\subsection{Proof of Proposition~\ref{prop:PDErepresentation}}
\label{sec:proofPDE}


We recall that  $(\lambda^*(t), \sigma^*(t),A^*(t), C^*(t))$ is the solution to
    the McKean--Vlasov equation~\eqref{eq:McKeanVlasov}.
We start by deriving the equation for $I(t, \cdot)$. Let us compute,
for some test function $\varphi$,
\begin{equation*}
    \E\big[ \varphi(A^*(t)) \indic_{\{C^*(t) = I\}} \big] 
    = \sum_{k \ge 0} \E\big[ \indic_{\{K^*(t) = k, C^*_{k} = I\}} \varphi(t-\tau^*_{k}) \big].
\end{equation*}
We have that
\[
    \{ K^*(t) = k \} \cap \{ C^*_{k} = I \} = \{ \tau^*_{k} \le t <
    \tau^*_{k}+T_{I,k} \} \cap \{ C^*_{k} = I \}.
\]
For $k = 0$, by our choice of initial condition, see [IC] in Section~\ref{SS:model},
\[
    \E\big[ \varphi(A^*(t)) \indic_{\{C^*(t) = I, K^*(t)=0\}} \big] 
    = I_0 \int_0^\infty h_I(a) \varphi(t+a) \exp\Big( -\int_a^{t+a} \mu_I(u) \diff u \Big) \diff a.
\]
For $k \ge 1$, $T_{I,k}$ is independent of $\tau^*_{k}$ and $C^*_{k}$
so that
\begin{align*}
    \E\big[ \varphi(A^*(t)) &\indic_{\{C^*(t) = I, K^*(t) \ge 1\}} \big] \\
    &= \sum_{k \ge 1} \E\Big[ \indic_{\{\tau^*_{k} \le t, C^*_k = I\}}
    \varphi(t-\tau^*_{k}) \exp\Big(- \int_0^{t-\tau^*_{k}} \mu_I(u) \diff u \Big) \Big]\\
    &= \E\Big[ \int_{[0,t]} \varphi(t-a) \exp\Big(- \int_0^{t-a} \mu_I(u) \diff u \Big)
    P_I(\diff a) \Big]
\end{align*}
where $P_I$ is the point process of infection times, which is the random
measure on $[0, \infty)$ defined as 
\[
    P_I(B) = \sum_{k \ge 1} \indic_{\{\tau^*_{k} \in B\}} \indic_{\{C^*_{k} = I\}}.
\]
Since infections occur at rate $\Lambda^*(t) \sigma^*(t)$ for $t \ge 0$,
the density of the intensity measure of $P_I$ is $\Lambda^*(t)
\E[\sigma^*(t)]$ so that  
\begin{equation*}
    \E\big[ \varphi(A^*(t)) \indic_{\{C^*(t) = I\}} \big] 
    = \begin{multlined}[t]
    \int_0^t \Lambda^*(a) \Sigma^*(a) \varphi(t-a) \exp\Big(- \int_0^{t-a} \mu_I(u) \diff u \Big)
    \diff a \\
    + I_0 \int_0^\infty h_I(a) \varphi(t+a) \exp\Big(- \int_a^{t+a} \mu_I(u) \diff u \Big)
    \diff a
    \end{multlined} \\
\end{equation*}
with $\Sigma^*(t) = \E[\sigma^*(t)]$. This shows that the density of 
$A^*(t)$ on $\{C^*(t) = I\}$ is 
\[
    \forall a \le t,\quad I(t,a) = \Lambda^*(t-a) \Sigma^*(t-a) \exp\Big(- \int_0^a \mu_I(u) \diff u \Big)
\]
and 
\[
    \forall a \ge t,\quad I(t,a) = I_0 h_I(a-t) \exp\Big(- \int_{a-t}^a \mu_I(u) \diff u \Big)
\]
which is the weak solution to
\begin{align*}
    \partial_t I(t,a) + \partial_a I (t,a)&= -\mu_I(a) I(t,a)\\
    I(t,0) &= \Lambda^*(t) \Sigma^*(t) \\
    I(0,a) &= I_0 h_I(a).
\end{align*}
We obtain the first part of the PDE limit \eqref{eq:main} ($\Sigma^*$ will be identified at the end of the proof). 
We now turn to the density of susceptible individuals. As previously
\[
    \{ K^*(t) = k \} \cap \{ C^*_{k} = S \} = \{ \tau^*_{k} \le t <
    \tau^*_{k}+T_{V,k} \wedge Z^*_{k} \} \cap \{ C^*_{k} = S \},
\]
so that for $k = 0$, recalling our initial condition [IC] and the
expression for the reinfection time $Z^*_{k}$ given by \eqref{eq:reinfectionTime} with a  $\Lambda^*$ instead of $\Lambda^N$ yields
that 
\begin{multline*}
    \E\big[ \varphi(A^*(t)) \indic_{\{C^*(t) = S, K^*(t) = 0\}} \big] \\
    = (1-I_0) \int_0^\infty h_S(a) \varphi(t+a) e^{- \int_a^{t+a} \mu_V(u) \diff u}
        \E\Big[e^{-\int_a^{t+a} \Lambda^*(u-a) \sigma(u) \diff u}\Big] \diff a.
\end{multline*}
Indeed, from \eqref{eq:reinfectionTime} and by independence, we notice that 
\begin{align*}
    \P(\min(T_{V,0}, Z^*_{0})>t+a \mid T_{V,0}>a)&=\P(T_{V,0}>t+a \mid T_{V,0}>a)\P( Z^*_{0}>t+a)\\
    &=e^{- \int_a^{t+a} \mu_V(u) \diff u}\,\P\left(\int_0^{t+a}\Lambda^*(-a+u)\sigma_0(u)\diff u<E_{0}\right)\\
      &=e^{- \int_a^{t+a} \mu_V(u) \diff u}  \E\Big[e^{-\int_a^{t+a} \Lambda^*(u-a) \sigma(u) \diff u}\Big].
\end{align*}
Similarly, for $k \ge 1$, using the independence of the various variables
\begin{multline*}
    \E\big[ \varphi(A^*(t)) \indic_{\{C^*(t) = S, K^*(t) \ge 1\}} \big] \\
    = \E\Big[ \int_{[0,t]} \varphi(t-a) e^{- \int_0^{t-a} \mu_V(u) \diff u}
        \E\Big[e^{-\int_0^{t-a} \Lambda^*(a+u)\sigma(u) \diff u}\Big]
    P_S(\diff a) \Big]
\end{multline*}
with the point process $P_S$ defined on $(0, \infty)$ as 
\[
    P_S(B) = \sum_{k \ge 1} \indic_{\{\tau^*_{k} \in B\}} \indic_{\{C^*_{k} = S\}}.
\]
The intensity of $P_S$ has a density that we denote by $p_S$. Recall
equation \eqref{eq:biasExpectation}, the next step is to note that 
\begin{equation} \label{eq:derivativeAttackRate}
    \E\Big[ e^{-\int_0^a \Lambda^*(t-a+u) \sigma(u) \diff u} \Big] = 
    e^{-\int_0^a \Lambda^*(t-a+u) \E_{t-a+u,u}[\sigma(u)] \diff u}.
\end{equation}
This is equivalent to showing that
\[
    \int_0^a \Lambda^*(t-a+u) \E_{t-a+u,u}[\sigma(u)] \diff u
    = - \log \E\Big[ e^{-\int_0^a \Lambda^*(t-a+u) \sigma(u) \diff u} \Big].
\]
Applying the operator $\partial_t + \partial_a$ to both sides leads to 
\[
    \Lambda^*(t) \E_{t,a}[\sigma(a)] = 
    \Lambda^*(t) \E\Big[ \sigma(a) e^{-\int_0^a \Lambda(t-a+u) \sigma(u) \diff u} \Big]
    \: \Big/\:\E\Big[ e^{-\int_0^a \Lambda(t-a+u) \sigma(u) \diff u} \Big]
\] 
and we recover the expression for $\E_{t,a}$, see \eqref{eq:biasExpectation}.

Therefore, combining the previous expressions,
\begin{align*}
    \E\big[ \varphi(A^*(t)) &\indic_{\{C^*(t) = S\}} \big]\\
    &= (1-I_0) \int_t^\infty h_S(a-t) \varphi(a) e^{- \int_{a-t}^a \mu_V(u) \diff u}
        \E\Big[e^{-\int_{a-t}^a \Lambda^*(t-a+u) \sigma(u) \diff u}\Big] \diff a \\
    &+ \int_0^t \varphi(a) e^{- \int_0^a \mu_V(u) \diff u-\int_0^a \Lambda^*(t-a+u) \E_{t-a+u,u}[\sigma(u)] \diff u}
        p_S(t-a) \diff a
\end{align*}
and we recover the weak solution of 
\begin{align*}
    \partial_t S(t,a) + \partial_a S(t,a) &= -\mu_V(a)S(t,a) - \Lambda^*(t) \E_{t,a}[\sigma(a)] S(t,a)\\
    S(t,0) &= p_S(t)\\
    S(0, a) &= (1-I_0) h_S(a).
\end{align*}
Our last task is to compute $\Lambda^*(t)$, $\Sigma^*(t)$, and 
$p_S(t)$. For the latter quantity, by construction of the process,
for any $t \ge 0$,
\[
    \E\big[ P_S([t, t+\diff t]) \mid A^*(t), C^*(t) \big] = \left(
    \mu_I(A^*(t)) \indic_{\{C^*(t) = I\}} + \mu_V(A^*(t)) \indic_{\{C^*(t) =
    S\}}\right) \diff t.
\]
Therefore
\[
    \forall t \ge 0, \quad p_S(t) = \int_0^\infty \mu_I(a) I(t,a) \diff a
    + \int_0^\infty \mu_V(a) S(t,a) \diff a.
\]
Similarly, by conditioning on $A^*(t)$, 
\[
    \forall t \ge 0,\quad 
    \Lambda^*(t) = \E[\lambda^*(t)] = \int_0^\infty I(t,a) \E\big[
    \lambda(a) \mid T_I > a\big] \diff a
\]
and 
\begin{align*}
    \forall t \ge 0,\quad
    \Sigma^*(t) &= \E[\sigma^*(t)] 
    = \int_0^\infty S(t,a) \E[\sigma(a) \mid T_V > a, Z > a] \diff a \\
    &= \int_0^\infty S(t,a) \E_{t,a}[\sigma(a)]\diff a.
\end{align*}


\section{Summary and discussion}\label{sec:discussion}

\paragraph{Summary.}
In this work we have proposed an individual-based model to study the
effect of recurrent vaccination on the establishment of an endemic
equilibrium, in a population with waning immunity. Our model incorporates
memory effects both for the transmission rate during an infection and for the
subsequent immunity, and takes into account the stochasticity at the
individual level for these two processes. By deriving the large
population size limit of the model and analysing its equilibria, we have
obtained a simple criterion for the existence of an endemic equilibrium.
This criterion depends jointly on the shape of the rate of immunity loss
and on the distribution of the time between two booster doses. In other
words, in the context of recurrent vaccination and waning immunity, 
what drives the result of a vaccination-policy is a combination of the
efficiency of the vaccine itself at blocking transmissions, and of the
way in which booster doses are distributed in the population. The
expression we obtain relates directly to the average immunity level
maintained by vaccinating recurrently the population, which is a relation
that we expect to hold for a broad class of models with similar
characteristics.

One general public health conclusion that we can draw from our work is
that, for the same average number of vaccine doses available, vaccination
strategies where the time between booster shots are more evenly spaced
(at the individual level) perform better at blocking transmissions.
A similar conclusion was reached recently by \cite{ElKhalifi2024}
for a related model.
Intuitively, irregularly spread booster doses lead to some longer time
intervals without vaccination, and the resulting high susceptibility
allows the disease to spread more efficiently. Deriving further
conclusions from our model would require to add some restrictions on the
distributions of $T_V$ and $\sigma$ that would reflect the
characteristics of a particular disease and vaccine.

Finally, we have studied two specific situations in more details. First,
we have computed an expression for the critical fraction of the
population required to adhere to the vaccination policy to eradicate the
disease (see \eqref{eq:criticalFraction}). This expression is reminiscent
of a well-known threshold for preventing an endemic state with an
imperfect vaccine \citep{anderson1985vaccination}. In the context of
recurrent vaccination, the efficiency of the vaccine is replaced by the
average susceptibility obtained by vaccinating individuals in the absence
of disease. Second we have studied the consequences of uneven vaccine
access in a population, and concluded that fair vaccine allocation is the
optimal strategy to prevent endemicity (see Proposition~\ref{prop:fairAllocation}).

\paragraph{Model assumptions.}
Our model is formulated in terms of infectiousness and susceptibility,
which are two phenomenological quantities that result from the complex
interaction between the pathogen and the host immune system. If this
interaction were modeled explicitly as in many existing works on viral
dynamics \citep{heffernan2008host, heffernan2009implications, goyal2020potency,
neant2021modeling}, infectiousness would relate to the viral load, and
susceptibility to the level of immune cells or circulating antibodies. 
Since we have left the susceptibility and infectiousness be general
random functions, our model should encompass many possible such host-pathogen
models. There are two assumptions that we have made about $\lambda$ and
$\sigma$ that could be easily relaxed mathematically, but would lead to a
more complicated model. First, we assumed that the susceptibility
curve following an infection is independent of the infectiousness curve
during that infection. A typical situation where this assumption would
fail is if a more severe infection leads both to a larger infectiousness
and to a higher level of immunity (and thus a lower susceptibility).
Second, we assumed that infection and vaccination lead to the same
susceptibility in distribution. We expect a law of large number
similar to Theorem~\ref{thm:lawLargeNumbers} to hold if these
assumptions are relaxed, with a similar criterion for the existence of
an endemic equilibrium and mild modifications of the limit equations. 
However, our mathematical results rely crucially on the strong assumption
that individuals (and thus their immune system) keep no memory of past
infections or vaccinations: at each reinfection or vaccination, the
subsequent infectiousness and susceptibility are sampled independently
and according to the same law. In particular, the expression
\eqref{eq:sigma} for the herd immunity threshold follows from the fact
that vaccinations form a renewal process, which is a consequence of this
absence of memory from past vaccinations. Relaxing this assumption would
require a completely different approach to our problem. Nonetheless, the
key quantity in our model is the stationary susceptibility of a typical
individual, obtained by letting an individual get vaccinated only for a
long period of time. It might be the case that other models displaying a
similar stationary behavior have the same qualitative properties as the
one investigated here.

The persistence of a disease requires a continual replenishment of
susceptible individuals to sustain the epidemic. In our
model, this influx of susceptibility comes exclusively from waning
immunity. Two other important causes for an increase in
susceptibility that we have neglected are the birth of new individuals
with no immunity and the pathogen evolution to escape immunity.
We expect that, as long as the population size is stable and newborns
start being vaccinated rapidly, demographic effects (that we have
neglected by considering a closed population) should not impact our
conclusions to a large extent. The key quantity that controls the
establishment of an endemic equilibrium in our model is the level
of population immunity in the absence of disease, which should be mostly
driven by vaccination if the typical time between two vaccine doses is
small compared to the lifetime of individuals. Taking into account
pathogen evolution is, however, a more challenging task that would
require further investigation and modeling. Though, note that a model
structured by time-since-recovery similar to the one we consider here
has been proposed to study the increase in susceptibility due to
antigenic drift in influenza strains \citep{pease1987evolutionary}.

\paragraph{Discussion.} 
In the second half of our work, we have used the endemic threshold $1/\Sigma$ to
quantify the efficiency of a given vaccination policy. This criterion has
the advantage of having a clear interpretation (in terms of the average
level of susceptibility maintained by vaccination), of being easy to
compute and of depending only on a few average quantities of the model:
the basic reproduction number $R_0$, the expected susceptibility at a
given time $\E[\sigma(a)]$, and the distribution of $T_V$. Another
interesting indicator of the impact of a vaccination policy is the
so-called endemic level, defined as the prevalence of the disease at the
endemic equilibrium. Ultimately, it is this endemic level that public
health measures try to control, to reduce the burden of the disease in
the population. In our model, when an endemic equilibrium exists, the
endemic level is given by $x \E[T_I] / R_0$, where $x$ solves $F_\e(x) =
R_0$ as in Proposition~\ref{prop:endemicFunction}. The endemic level is
therefore only implicitly defined, which makes it more complicated to
study both analytically and numerically. Investigating the impact of the
vaccination policy on the endemic level, though important, would
therefore require further work, and the conclusions reached in
Section~\ref{SS:endemicity_2pop} could be altered by using this endemic
level as a criterion for the efficiency of vaccination instead of the
endemic threshold. Note that the question of the impact of the way
immunity is waning on the endemic level has been the subject of a recent
study \citep{khalifi2022extending}.

The simplest epidemic models consider the
spread of a disease in a population made of identical individuals, that
are mixing homogeneously: individuals are equally susceptible to the
disease, equally infectious once infected, and contacts are equally
likely to occur between any pair of individuals in the population
\citep[Part I]{britton2019stochastic}. 
Many works have studied the epidemiological consequences of relaxing
these assumptions, to account for some of the heterogeneity which is
observed in human populations \citep{britton2020mathematical,
brauer2008epidemic, magal2016final, andreasen2011final, david2018epidemic}. 
In a similar way, we have added some heterogeneity to our model in
Section~\ref{S:twoGroups} by assuming that the population is subdivided
into a finite number of groups, that contacts between groups are
heterogeneous and that individuals in different groups get vaccinated
according to different distributions of $T_V$. Since our focus is the
impact of inhomogeneous vaccination on endemicity, we have assumed that
all groups have the same activity level ($\Gamma \times \diag(p)$ is a
multiple of a stochastic matrix), and that they sample their
infectiousness and susceptibility curves from the same distribution. Our
model could be easily extended to allow the distribution of the
infectiousness and susceptibility curves to depend on the group, and to
general contact matrices $\Gamma$. Using the same heuristic arguments as
in Section~\ref{SS:endemicHeterogeneous}, we can derive a criterion for
the existence of an endemic equilibrium in terms of the leading
eigenvalue of a next-generation matrix similar to \eqref{eq:nextGenMatrix}. 
However, although it is possible to derive such an expression, the joint
effect of heterogeneous infectiousness, susceptibility, contact rates and
vaccination rates on this criterion is extremely complex, but
would be a very interesting avenue for future work.

Finally, following the tradition of classical epidemiology models, we
have considered the groups as being fixed during the course of the
epidemic. Although this assumption might be realistic if groups model
``physical'' heterogeneities (age classes on a short time-scale,
spatial locations), it becomes simplistic when groups model human behavior
(compliance to public health measures, vaccine hesitancy). In the latter
situation, the group to which an individual belongs can possibly change
and is influenced by many factors, including perceived risks of vaccine
adverse events, disease prevalence, or available information.
Modeling such effects appropriately is a challenging task that is the
focus of behavioral epidemiology \citep{Bauch2005, Donofrio2011,
Lupica2020, Manfredi2013}. Incorporating such effects in the
current model is an interesting avenue for future work.

It is interesting to compare our expression for the endemic threshold to
that recently obtained in \cite{forien22}, for a similar model but
without vaccination. In \cite{forien22} it is shown that, in the absence
of vaccination and using our notation, an endemic equilibrium exists if
and only if
\begin{equation} \label{eq:endemicThem}
    \frac{1}{R_0} \E\Big[ \frac{1}{\sigma_*} \Big] < 1,
\end{equation}
where $\sigma_* \coloneqq \lim_{a \to \infty} \sigma(a)$. The corresponding
expression with vaccination that we have obtained is
\begin{equation} \label{eq:endemicVaccine}
    \frac{\E[T_V]}{R_0 \E\Big[ \int_0^{T_V} \sigma(a) \diff a \Big]} < 1.
\end{equation}
By letting $T_V \to \infty$, we expect that our model converges to the
model considered in \cite{forien22} where no vaccination is taken into
account. However, in the limit $T_V \to \infty$ our expression for
the endemicity criterion becomes
\begin{equation} \label{eq:endemicUs}
    \frac{1}{R_0 \E[\sigma_*]} < 1.
\end{equation}
Note the surprising discrepancy between \eqref{eq:endemicThem} and
\eqref{eq:endemicUs}. This apparent contradiction can be resolved by
noting that both expressions are specific cases of a more general
formula. Let $\zeta(u)$ denote the susceptibility of a typical individual
at age-of-infection $u$, that is, $u$ unit of time after its last
infection, regardless of how many times it has been vaccinated since then.
Then, mimicking the computations of Section~\ref{S:equilibria} would
suggest that the correct threshold for the existence of an endemic
equilibrium is given by
\begin{equation} \label{eq:endemicGeneral}
    \frac{1}{R_0} \E\Big[ \lim_{a \to \infty} \frac{1}{\frac{1}{a}\int_0^a \zeta(u) \diff u}
    \Big] < 1,
\end{equation}
provided that the limit in the expectation exists. In the absence of
vaccination, $\zeta(u) = \sigma(u)$ and we recover \eqref{eq:endemicThem}. 
In the presence of vaccination, letting $\sigma_i$ and $T_i$ be i.i.d.,
we have 
\[
    \forall u \ge 0,\quad \zeta(u) = \sigma_i(u - (T_1+\dots+T_i))
    \indic_{\{T_1 + \dots + T_i \le u < T_1 + \dots + T_{i+1}\}}
\]
and classical results on renewal processes show that we recover
\eqref{eq:endemicUs}. We believe that \eqref{eq:endemicGeneral} should
give the correct threshold for the existence of an endemic equilibrium for
a broader class of models.

\bookmarksetup{startatroot}
\belowpdfbookmark{Acknowledgements}{Acknowledgements}

\subsection*{Acknowledgements}

We would like to thank two reviewers for taking the time and effort
necessary to review the manuscript. We sincerely appreciate all valuable
comments and suggestions, which helped us to improve the quality of the
manuscript.

\medskip
\noindent
The project was sponsored by Mathematics for Public Health (MfPH)
initiative, in Canada, involving the Fields Institute, the Atlantic
Association for Research in Mathematical Sciences (AARMS), the Centre de
Recherches Mathématiques (CRM), and the Pacific Institute for
Mathematical Sciences (PIMS). FFR acknowledges financial support from
MfPH, the AXA Research Fund, the Glasstone Research Fellowship, and
thanks Magdalen College Oxford for a senior Demyship. AC acknowledges
Canada's Natural Sciences and Engineering Research Council (NSERC) for
funding (RGPIN-2019-07077), the AXA Research Fund and the SCOR Foundation for Research. HG acknowledges
Canada's Natural Sciences and Engineering Research Council (NSERC) for
funding (RGPIN-2020-07239).

\subsection*{Data availability statement}
Data sharing not applicable to this article as no datasets were generated
or analysed during the current study.

\appendix

\section{Numerical simulations}\label{sec:numerical:simulations}

\subsection{Model parametrization}\label{sec:model:param}

In all simulations, we assumed that the laws 
$\mathcal{L}_\lambda$, $\mathcal{L}_\sigma$ and $\mathcal{L}_V$ 
have the following form.

The law $\mathcal{L}_\lambda$ depends on four parameters: two shape
parameters $\kappa_I, \kappa'_I$, one scale parameter $\theta_I$ and one
parameter $R_0$ for the total mass. Then, we define two independent random
variables
\[
    T_I \sim \mathrm{Gamma}(\kappa_I + \kappa'_I, \theta_I) \coloneqq
    \frac{x^{\kappa_I+\kappa'_I-1}\,e^{-x/\theta_I}}{\Gamma(\kappa_I+\kappa'_I)
    \theta_I^{\kappa_I+\kappa'_I}} \diff x, 
    \qquad 
    M \sim \mathrm{Exponential}(R_0) \coloneqq \frac{e^{-x/R_0}}{R_0} \diff x
\]
and we set
\[
    \forall a < T_I,\quad \lambda(a) = 
    \frac{M}{T_I \mathrm{B}(\kappa_I,\kappa'_I)} 
    \big(\tfrac{a}{T_I}\big)^{\kappa_I-1} \big(1-\tfrac{a}{T_I}\big)^{\kappa'_I-1}
\]
where $\mathrm{B}((\kappa, \kappa')$ is the Beta function. In words, $\lambda$
has the shape of the density of a Beta distribution, stretched by a
random factor $T_I$, and with random total mass $M$. Note that we do have
\[
    \E\Big[ \int_0^{T_I} \lambda(a) \diff a \Big] = R_0.
\]
This choice of parametrization was chosen so that $\E[\lambda(a)]$ is
easy to compute by standard properties of Beta and Gamma distributions.

\begin{table}
    \centering
    \def\arraystretch{1.2}
    \begin{tabular}{c c c}
        \toprule
        Parameter & Description & Value \\

        \midrule
        
        $\kappa_I, \kappa_I'$ & Shape parameters of $\lambda$ & $3$ \\
        $\theta_I$ & Scale parameter of $\lambda$ & $2$ \\

        $\kappa_V$ & Shape parameter of $T_V$ & $4$ \\
        $\theta_V$ & Scale parameter of $T_V$ & $3$ \\

        $\kappa_\sigma$ & Shape parameter of $\sigma$ & $5$ \\
        $\theta_\sigma$ & Scale parameter of $\sigma$ & $2$ \\

        $h A^*$ & Maximal age & $150$ \\

        \bottomrule
    \end{tabular} 
    \caption{Default parameter values for the simulations.}
    \label{tab:parametersValue}
\end{table}

In all simulations, except those of
Figure~\ref{fig:variationEndFunction}, the susceptibility curve $\sigma$
is deterministic and depends on two parameters: a shape parameter
$\kappa_\sigma$ and a scale parameter $\theta_\sigma$. Let $G(x; \kappa,
\theta)$ be the cumulative distribution function of a
$\mathrm{Gamma}(\kappa, \theta)$ random variable, evaluated at $x$, that
is,
\[
    \forall x \ge 0,\quad G(x; \kappa, \theta) = \P(X \le x),
    \qquad 
    X \sim \mathrm{Gamma}(\kappa, \theta).
\]
We assume that 
\[
    \forall a \ge 0,\quad \sigma(a) = G(a; \kappa_\sigma, \theta_\sigma).
\]
In the simulations of Figure~\ref{fig:variationEndFunction} only, we
assume that $\sigma$ is stochastic. More precisely, we obtain a random
$\sigma$ by letting the parameters of the Gamma distribution be
themselves random (and Gamma distributed with mean $\kappa_\sigma$ and
$\theta_\sigma$), that is,
\begin{equation} \label{eq:suscep_random}
    \forall a \ge 0,\quad \sigma(a) = G(a; K_\sigma, \Theta_\sigma),
    \qquad K_\sigma \sim \mathrm{Gamma}(2,
    \tfrac{\kappa_\sigma}{2}),
    \qquad \Theta_\sigma \sim \mathrm{Gamma}(2,
    \tfrac{\theta_\sigma}{2}).
\end{equation}

The law of the vaccination duration $T_V$ also depends on one shape
parameter $\kappa_V$ and one scale parameter $\theta_V$. We simply assume
that $T_V$ follows a Gamma distribution with these parameters:
\[
    T_V \sim \mathrm{Gamma}(\kappa_V, \theta_V).
\]
For the heterogeneous model, we use a similar parametrization. We suppose
that for each $i \in \{1, \dots, L\}$ we have
\[
    T_i \sim \mathrm{Gamma}(\kappa_i, \theta_i)
\]
where $\kappa_i, \theta_i$ are the shape and scale parameters of $T_i$.

\begin{figure}
    \centering
    \includegraphics[width=\textwidth]{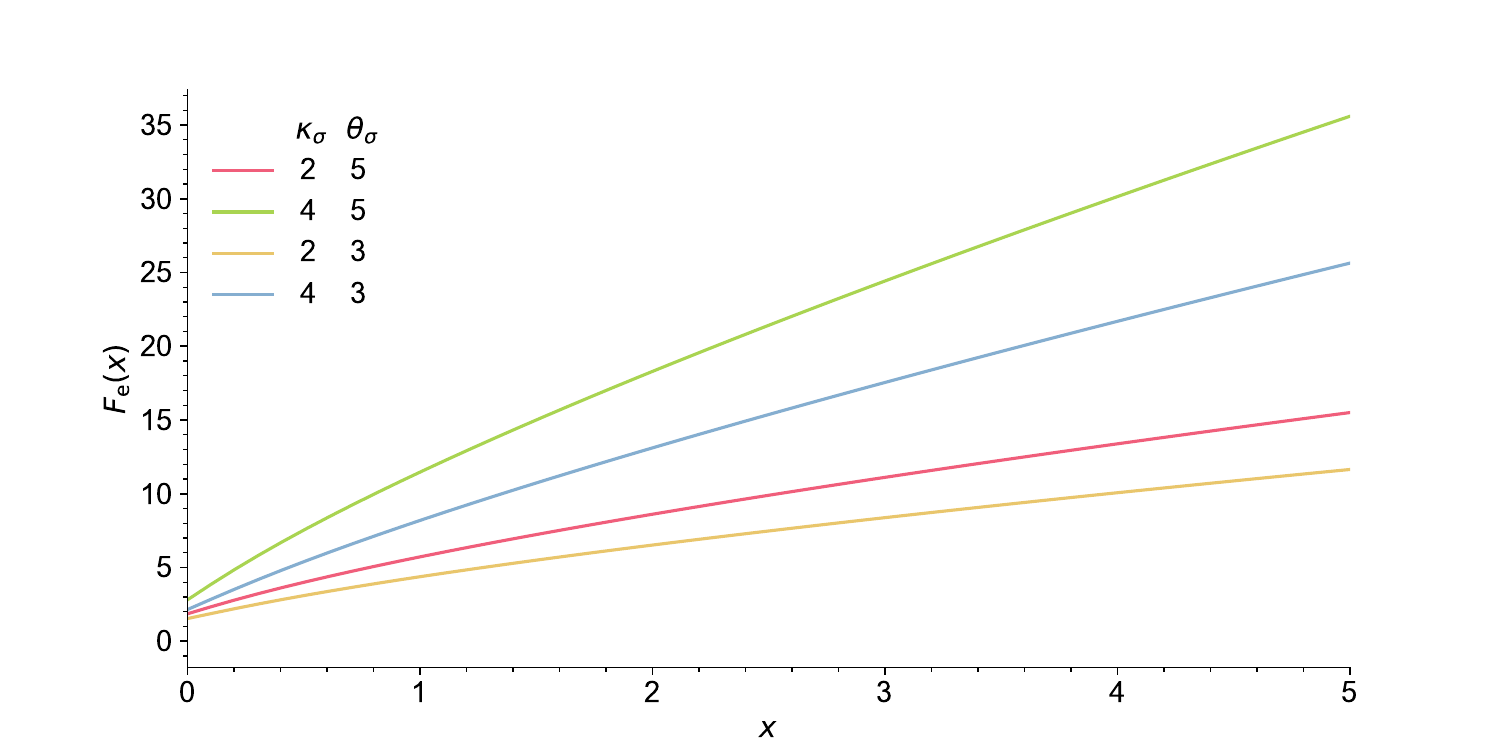}
    \caption{Numerical approximation of the function $F_\e$ in
    \eqref{eq:endemicity}. The susceptibility $\sigma$ is given by
    \eqref{eq:suscep_random}, where the parameters $\kappa_\sigma$ and
    $\theta_\sigma$ are given in the legend. The parameters of $T_V$ are
    given in Table~\ref{tab:parametersValue} (Appendix~\ref{sec:numerical:simulations}), and we assumed $\E[T_I] = 0$.} 
    \label{fig:variationEndFunction}
\end{figure}

\subsection{Initial condition}

We selected an initial conditions close to the endemic equilibrium to
speed up the convergence to it, and make the transient behavior of the
system shorter. Precisely, for all simulations we used
\[
    \forall a \ge 0,\qquad 
    h_I(a) = \exp\Big( - \int_0^a \mu_I(u) \diff u \Big) / \E[T_I],
    \qquad
    h_S(a) = \exp\Big( - \int_0^a \mu_V(u) \diff u \Big) / \E[T_V].
\]

For the heterogeneous model, we set the initial age structure in each
group as above, using the vaccination time distribution $T_i$ of the
corresponding group for $h_S$.

\begin{table}
    \centering
    \def\arraystretch{1.2}
    \begin{tabular}{c c c}
        \toprule
        Parameter & Description & Value \\

        \midrule

        $\kappa_1$, $\kappa_2$, $\kappa_3$ & Shape parameters of $T_1$, $T_2$, $T_3$ & $10$, $5$, $6$ \\

        $\theta_1$, $\theta_2$, $\theta_3$ & Scale parameters of $T_1$, $T_2$, $T_3$ & $2$, $4$, $1.5$ \\

        $p_1$, $p_2$, $p_3$ & Proportions of individuals in each groups & $0.1$, $0.3$, $0.6$ \\
        
        $\Gamma$ & Contact rates matrix &  
        $\begin{pmatrix}
            1.5 & 1 & 0.2 \\
            1 & 2 & 1.5 \\
            0.2 & 0.5 & 3
        \end{pmatrix}$ \\

        \bottomrule
    \end{tabular} 
    \caption{Additional parameter values for the simulations of the heterogeneous
    model.}
    \label{tab:parametersValueHeterogeneous}
\end{table}

\subsection{Numerical approximations of the PDE}\label{sec:apprendix:num:approx:pde}

We approximate the solution $(S(t,a), I(t,a);\, t,a \ge 0)$ of the PDE
\eqref{eq:main} on a finite lattice $\mathcal{G} = \{ (\eta i, \eta k);\, i \in
\{0, \dots, T^*\}, k \in \{0,\dots, A^*\} \}$ with a small discretization
step $\eta > 0$, using the method of characteristics.

The approximation is a vector $(I_{i,k}, S_{i,k};\, (i,k) \in \mathcal{G})$ defined
inductively as follows. We let the time boundary condition be
\[
    \forall k \in \{0,\dots, K^*\},\qquad
    I_{0,k} = I_0 \frac{h_I(\eta k)}{\eta \sum_{k=0}^{K^*} h_I(\eta k)},
    \qquad
    S_{0,k} = (1-I_0) \frac{h_S(\eta k)}{\eta \sum_{k=0}^{K^*} h_S(\eta k)},
\]
Moreover, we define the approximation of the force of infection as
\[
    \forall i \le T^*,\quad \Lambda_i = \eta \sum_{k=0}^{A^*} \E[\lambda(\eta k)] I_{i,k}.
\]
For the age boundary condition we let for $i \in \{0, \dots, T^*-1\}$
\[
    I_{i+1, 0} = \Lambda_i \cdot \Big( \eta \sum_{k=0}^{A^*} \sigma(\eta k) S_{i,k} \Big),
    \qquad
    S_{i+1, 0} = \eta \sum_{k=0}^{A^*} \big( \mu_I(\eta k) I_{i,k} + \mu_V(\eta k) S_{i,k} \big).
\]
Finally, for $i \in \{0, \dots T^*-1\}$ and $k \in \{0, \dots, K^*-1\}$ we define
\begin{gather*}
    I_{i+1,k+1} = I_{i, k} \big( 1 - \eta \mu_I(\eta k) \big) +
    \indic_{\{k+1 = K^*\}} I_{i, K^*} \big( 1 - \eta \mu_I(\eta K^*) \big) \\
    S_{i+1,k+1} = S_{i, k} \big( 1 - \eta ( \mu_V(\eta k) + \Lambda_i \sigma(\eta k) ) \big)
    + \indic_{\{k+1 = K^*\}} S_{i, K^*} \big( 1 - \eta ( \mu_V(\eta K^*) + \Lambda_i \sigma(\eta K^*) ) \big).
\end{gather*}
Note that the second term in each expression corresponds to a reflecting
age boundary at $k = K^*$.

The solution of the PDE \eqref{eq:mainHeterogeneous} for heterogeneous contacts is
approximated using a straightforward adaptation of the above scheme.

\bibliographystyle{plainnat}
\bibliography{biblio.bib}

\end{document}